\documentclass[11 pt]{article}

\usepackage[english]{babel}
\usepackage{xr}

\usepackage[letterpaper,top=2cm,bottom=2cm,left=3cm,right=3cm,marginparwidth=1.75cm]{geometry}

\usepackage{amsmath, amsthm, amsfonts,enumerate, mathtools, amssymb}

\usepackage{graphicx}
\usepackage[colorlinks=true, allcolors=blue]{hyperref}
\usepackage[round]{natbib}
\usepackage{bm} 
\usepackage{bbm}
\usepackage{cases}
\usepackage{subcaption}
\usepackage{accents}
\usepackage{comment}
\usepackage{mathtools}
\DeclarePairedDelimiter{\ceil}{\lceil}{\rceil}

\newcommand{\bX}{{\bf X}}

\newcommand{\bY}{{\bf Y}}
\newcommand{\bZ}{{\bf Z}}
\newcommand{\So}{\Sigma_{0}}
\newcommand{\Sa}{\Sigma_{1}}
\newcommand{\Sk}{\Sigma_{k}}

\newcommand{\bmu}{\boldsymbol{\mu}}

\newcommand{\Exp}{\mathbb{E}}
\newcommand{\Prob}{\mathbb{P}}
\newcommand{\tr}{\mathrm{tr}}
\newcommand{\LM}{\lambda_{\max}}
\newcommand{\Lm}{\lambda_{\min}}

\newcommand{\Var}{\mathrm{Var}}

\newcommand{\kp}{k^{\prime}}

\newcommand{\delhatrpe}{\widehat{D}^{\mathrm{RPE}}}

\newcommand{\delhatrp}{\widehat{D}^{\mathrm{RP}}}

\newlength{\dhatheight}

\newcommand{\Ao}{\hyperlink{A2}{A2}}
\newcommand{\At}{\hyperlink{A1}{A1}}
\newcommand{\Ath}{\hyperlink{A3}{A3}}
\newcommand{\Af}{\hyperlink{A4}{A4}}

\newtheorem{theorem}{Theorem}
\newtheorem{lemma}{Lemma}
\newtheorem{example}{Example}
\newtheorem{remark}{Remark}

\title{\bf Ultrahigh-dimensional Quadratic Discriminant Analysis Using Random Projections}
\author{Annesha Deb$^1$, Minerva Mukhopadhyay$^{1,2}$ and Subhajit Dutta$^{1,3}$}
\date{\today}

\begin{document}
\maketitle

\noindent
{\it $^{1}$Department of Mathematics and Statistics,
Indian Institute of Technology Kanpur, Kanpur - 208016, UP, India.}\\
{\it $^2$Interdisciplinary Statistical Research Unit, Indian Statistical Institute, 203 B. T. Road, Kolkata – 700108, WB, India.}\\
{\it $^3$Applied Statistics Unit,
Indian Statistical Institute, 203 B. T. Road, Kolkata – 700108, WB,~India.}

\begin{abstract}
This paper investigates the effectiveness of using the Random Projection Ensemble (RPE) approach in Quadratic Discriminant Analysis (QDA) for ultrahigh-dimensional classification problems. Classical methods such as Linear Discriminant Analysis (LDA) and QDA are used widely, but face significant challenges in their implementation when the data dimension (say, $p$) exceeds the sample size (say, $n$). In particular, both LDA (using the Moore-Penrose inverse for covariance matrices) and QDA (even with known covariance matrices) may perform as poorly as random guessing when $p/n \to \infty$ as $n \to \infty$. The RPE method, known for addressing the curse of dimensionality, offers a fast and effective solution without relying on selective summary measures of the competing distributions. This paper demonstrates the practical advantages of employing RPE on QDA in terms of classification performance as well as computational efficiency. We establish results for limiting perfect classification in both the population and sample versions of the proposed RPE-QDA classifier, under fairly general assumptions that allow for sub-exponential growth of $p$ relative to $n$. Several simulated and gene expression data sets are analyzed to evaluate the performance of the proposed classifier in ultrahigh-dimensional~scenarios.


\end{abstract}

\section{Introduction}

Classification techniques are built using labeled data that contain multiple predictor variables along with their corresponding class labels. The data can be represented as $(\bm{X}, y)$, where $\bm{X}$ denotes the $p$-dimensional predictor variable and $y$ denotes the corresponding class label. The broad goal of discriminant analysis is to construct a class boundary in this $p$-dimensional space that optimally separates the competing populations.
To fix ideas, let us consider the gene expression data for $72$ Leukemia patients (available \sloppy{\href{https://file.biolab.si/biolab/supp/bi-cancer/projections/info/MLL.html}{here}})
% (available \href{https://file.biolab.si/biolab/supp/bi-cancer/projections/info/braintumor.html}{here}) 
with three types of Leukemia, each associated with $12533$ gene expressions. The primary objective is to classify a patient into one of these three types based on the gene expressions from the patient (a detailed analysis of this dataset has been provided later in Section \ref{sec:realdata-analysis}).

In high-dimensional settings, most existing classical classification methods encounter significant challenges. We either get a sub-optimal value for the misclassification probability (see, e.g., \cite{bickel2004some, li2015sparse}), or have to deal with excessive computational costs (which may even be infeasible in some scenarios). In particular, classical $\text{LDA}$ and $\text{QDA}$ involve the inverse of sample covariance matrices, which are singular in this setup as $p>n$. 
In this paper, we focus on a computationally efficient version of QDA, where given a class label $y$, the $p$-dimensional predictor $\bm{X}$ is assumed to follow a Gaussian distribution with different means and covariance matrices. 

Several approaches have been proposed in the literature for high-dimensional QDA which replace various functions of the unknown mean vectors and covariance matrices of the competing populations with their sample estimates.
%to optimize the misclassification rate. 
As noted by \cite{friedman1989regularized}, sample estimates of covariance matrices lead to biased estimates of their eigenvalues when the dimension grows with the sample size (i.e., $p \sim n$) and this may lead to inaccurate classification results.
To mitigate this problem, the author proposed regularized discriminant analysis, where the sample estimates of covariance matrices were replaced with their regularized variants.
\cite{dudoit2002comparison} utilized a diagonal discriminant approach to analyze high-dimensional gene expression data, specifically to address the inversion problem. 
\cite{bouveyron2007high} introduced a parameterization of the covariance matrices based on the assumption that high-dimensional data essentially lies in low-dimensional subspaces. In a similar approach, \cite{wu2019quadratic} implemented QDA assuming a simplified structure of the covariance~matrices. 

Under suitable sparsity assumptions on the covariance matrices and differences of means, \cite{li2015sparse} developed a sparse QDA (SQDA) classifier and derived its asymptotic optimality. Subsequently, \cite{fan2015innovated} proposed a two-step procedure called IIS-SQDA, which incorporates an innovative interactive screening (IIS) to reduce the dimensionality prior to applying SQDA. Assuming sparsity, \cite{jiang2018direct} formulated a direct approach to QDA (DA-QDA), by estimating the key quantities in the discriminant function of QDA instead of the population parameters separately. The authors further established convergence of the misclassification probability of DA-QDA to the optimal Bayes~risk.
Assuming a strongly spiked covariance model, \cite{AoYa2014, aoshima2019distance} proposed a distance based QDA~classifier.

In the context of dimension reduction, random projection (RP) is widely recognized as an efficient tool which is supported with a sound theoretical foundation. The main motivation stems from the well-known Johnson-Lindenstrauss (JL) lemma \citep{johnsonlemma} and its randomized variants (see, e.g., \cite{dasgupta2003elementary}). The JL lemma essentially states that random linear projection of two high-dimensional points onto an appropriate low-dimensional subspace preserves their interpoint distance with high probability, or in expectation. This property makes RP particularly effective for high-dimensional inference problems, where the preservation of interpoint distances is crucial.
In practice, high-dimensional data points are first projected onto a low-dimensional subspace using a suitable random matrix and subsequently, inference is performed on these projected samples. 
To limit the subjectivity of a single random matrix, the random projection ensemble (RPE) method aggregates results from multiple RPs.

RP and RPE have been successfully implemented in kNN classification (see, e.g, \cite{deegalla2006reducing, yan2019k}) and clustering (\cite{dasgupta1999learning, fern2003random, heckel2017dimensionality}), LDA and QDA based classification (\cite{durrant2013sharp, durrant2013random}), hypothesis testing (\cite{lopes2011more, RAPTT2016, ayyala2022covariance}), estimation of the precision matrix (\cite{marzetta2011random}), regression (\cite{klanke2008library, mukhopadhyay2020targeted, ahfock2021statistical}), sparse principal component analysis~\citep{gataric2020sparse}, etc. RP(E) techniques are quite popular due to their ease of application and significantly low computational cost.

In this paper, we explore the application of RPE in QDA for ultrahigh-dimensional settings. 
The primary contributions of this paper, along with their potential novelty relative to earlier work, is outlined in the following subsection.

\subsection{Our Contributions}

In this paper, we introduce an RPE-based QDA classifier, referred to as RPE-QDA. The method generates \( B \) random matrices of dimension \( d \times p \) from a suitable distribution with \( d \ll p \) and \( d < n \). Using each of these random matrices, the training data is projected onto a random \( d \)-dimensional subspace and a random QDA discriminant function is constructed in this subspace. The classification is performed based on an aggregation of $B$ random discriminant functions called the RPE-QDA discriminant function. A detailed description of RPE-QDA is provided in Section \ref{sec:RPE}.

%{\color{red} 
Under a set of reasonable assumptions on the means and covariance matrices of the competing populations (see Section \ref{assumptions}), we establish \textit{perfect classification} (i.e., the misclassification probability approaches zero as \( p \to \infty \)) for the proposed RPE-QDA classifier. These assumptions are sufficiently flexible to accommodate sub-exponential growth of \( p \) relative to \( n \) (i.e., \( p \sim \exp\{o(n)\} \)), no mean difference among the competing distributions and various covariance structures including the spiked \citep{johnstone2001distribution, aoshima2019distance} as well as identity-type matrices. 

To provide further insight to our method, we mathematically quantify the loss in employing the RPE in QDA compared to the classical QDA. 
Under appropriate assumptions, we establish an equivalence between the classical QDA discriminant function and the Kullback–Leibler (KL) divergence. 
% The intuition behind this correspondence is the following. 
% The KL divergence quantifies the divergence or separability between two probability distributions (see \cite{KL1951}). Specifically, the KL divergence of $P_k$ to $P_{\kp}$ is $\mathbb{E}_{\bm{X} \sim f_k}[\log(f_k(\bm{X})) - \log(f_{\kp}(\bm{X}))]$, while the discriminant function (excluding the prior terms) is $\log(f_k(\bm{X})) - \log(f_{\kp}(\bm{X}))$ (see equation \eqref{eqn:PDF}). If $\bm{X} \sim f_k$, this discriminant is expected to be positive for most values of $\bm{X}$. However, if the distributions are not well separated, the discriminant can frequently take negative values, leading to a higher probability of misclassification. A larger KL divergence increases the quantity $\mathbb{P}(\{\log(f_{k}(\bm{X})) - \log(f_{\kp}(\bm{X}))\} > 0)$,i.e., the probability that the discriminant function correctly identifies the true population of a new observation. Hence, under the Gaussianity assumption, a higher KL divergence corresponds to a lower Bayes error and therefore a greater discriminative power of the QDA rule. 
In contrast, the RPE-QDA discriminant function is shown to be equivalent to a scaled version of the KL divergence, with a scaling factor of $d/p$ (see Section \ref{PC_RPE-QDA} for details).
%, where $d$ is the reduced dimension. 
This reduction in scale introduced by the RPE technique is traded off by an improved computational efficiency.
We illustrate the performance of RPE-QDA through extensive simulations (see Section \ref{sec:simulation}) demonstrating its applicability
%compared to the summary statistics based high-dimensional generalizations of QDA, 
in a wide class of classification problems. 
Finally, we corroborate the proposed method by analyzing several multi-class gene expression data sets in Section \ref{sec:realdata-analysis}.  
%}

\vskip5pt
\noindent{\it Related work.}
\cite{cannings2017random} proposed a general RPE based classification method
and implemented the same on QDA, along with LDA and kNN classifiers. The approach of this paper differs from ours in two aspects. First, in the choice of random matrices, where they employed a cross-validation step to select a few \emph{good} matrices from a large collection of random matrices, but we simply simulate random matrices from appropriate distributions. Secondly, they aggregate on the decisions in the ensemble, while we aggregate the discriminant functions. 
From a theoretical perspective, they provide an upper bound on the difference between the expected test error of their proposed RPE classifier and the Bayes risk. 
%From a theoretical perspective, they provide an upper bound on the \emph{excess risk}, i.e., the difference between the expected test error of their proposed RPE classifier and the Bayes risk. 
However, this bound could not be stated explicitly for QDA.
%due to the non-availability of an analytic form of its \emph{excess risk}.
%{\color{red} 

On a related note, \cite{palias2023effect} derived an upper bound on the Bayes error of QDA under various projection schemes. This bound depends on the difference in class means and certain summary measures of the covariance matrices such as their traces and largest eigenvalues. When these quantities are identical across the two populations, the bound becomes uninformative (see the comparison with existing methods in Section \ref{PC_RPE-QDA}), although the RP-QDA method continues to perform well in practice.
% On a related note, \cite{palias2023effect} provide an upper bound on the Bayes error of QDA under different projection schemes.
% This upper bound depends on the mean difference and certain summary measures of the covariance matrices, such as their traces and highest eigenvalues. 
% If these measures are the same for the two competing populations, then the upper bound becomes uninformative {\color{red}(see ??? for details)}, although the RP-QDA method still performs well in such situations. 
Moreover, their theoretical results neither include the ensemble approach nor account for the sample version of the QDA~classifier.

%Generalizations of high-dimensional QDA proposed by \cite{AoYa2014, wu2019quadratic} also involve summary measures and face similar difficulties {\color{red}(see ??? for details).}
%In this aspect, RP-QDA has a strong potential and provides an alternative method for high-dimensional QDA (we have demonstrated this in Section \ref{sec:theory}).  
%}

\subsection{Notations} We now list down the key notations and conventions used in the paper below.
 \vspace{-.1 in}
\begin{enumerate}
\itemsep-0.5em
\item Let $a$ and $b$ be two real numbers, then $a\wedge b=\max\{a,b\}$ and $a\vee b= \min\{a,b\}$.

%\item If $\{p\}$ and $\{n\}$ be sequences of real numbers, then $p \gg n$ implies $n/p \rightarrow 0$ as $n \to \infty$.

\item Let $\{a_{n}\}$ and $\{b_{n}\}$ be sequences of real numbers, then $a_{n} \gg b_{n}$ implies $a_{n}/b_{n} \to \infty$,  and $a_{n} \lesssim b_{n}$ (equivalently, $b_{n}\gtrsim a_{n}$) implies that $a_{n}\leq c b_{n}$ for all sufficiently large $n$ for some $c>0$. Further, $a_{n} \sim b_{n}$ implies that both $a_{n} \lesssim b_{n}$ and $b_{n} \lesssim a_{n}$ hold.

\item If $\bm{X}$ is a $p$-dimensional  random vector, then $\bm{X} \sim Q$ (here $Q$ is a probability distribution) implies $\bm{X}$ is distributed as~$Q$.

\item Let ${A}$ be a square matrix. Then, $\lambda_{\min}({A})$ and $\lambda_{\max}({A})$ denote the minimum and maximum eigenvalues of ${A}$, respectively. 
\end{enumerate}

\vspace*{-0.2in}
\section{The Random Projection Approach for QDA}

Consider the problem of classifying a random variable $\bZ \in \mathbb{R}^{p}$ to one of $J$ possible classes, and $y$ denotes its class label with $y \in \{1, \ldots, J\}$. 
Let the $k$-th class be associated with the population $P_{k}$ for $k = 1, \ldots, J$.
Quadratic discriminant analysis (QDA) assumes the conditional distribution of $\bZ$ given $y$ as Gaussian, with the following specifications
$$\bZ \mid y=k \sim N_p(\bmu_{k}, \Sigma_{k}),$$
where $\bmu_{k} \in \mathbb{R}^{p}$ and $\Sigma_{k}$ is a $p \times p$ real positive definite matrix for $k = 1, \ldots, J$.
In other words, the probability distribution of the $k$-th population $P_{k}$ is $N_{p}(\bmu_{k}, \Sk)$ for $k = 1, \ldots, J$.
Unlike linear discriminant analysis (LDA), which assumes a fixed covariance matrix for all the $J$ populations, QDA allows both the parameters $\bmu_{k}$ and $\Sk$ to vary across populations. 

Fix $k_{0} \in \{ 1, \ldots, J\}$.
The Bayes decision rule (say, $\delta^B(\bZ)$) is given~by
\begin{equation} \delta^B(\bZ) = 
    k_{0} ~~ \text{if} \quad  {\arg\max}_{k \in \{1, \ldots, J\} }  \pi_{k} f_{k}(\bZ) = k_{0}, \label{eqn:Classif1}  
\end{equation}
where $\pi_{k} = \Prob (\bZ \in P_k)$ is the prior of occurrence of the $k$-th population with $\sum_{k} \pi_{k} = 1$, and $f_{k}$ is the density of the $k$-th population with $k \in \{1, \ldots, J\}$. 
Under Gaussianity, the Bayes decision rule associated with QDA can be equivalently expressed as follows:
\begin{equation} 
\delta^{\mathrm{QDA}}(\bZ) = 
    k_{0}  ~~ \text{if} ~  D_{k_{0}, k }(\bZ) > 0 ~~ \text{for all} ~~  k \in \{ 1, \ldots, J\} \setminus \{k_{0}\}.
    \label{eqn:Classif2}  
\end{equation} 
Here, $D_{\kp, k }(\bZ)$ is the log discriminant function of classes $k^{\prime}$ and $k$, which can be expressed~as
\begin{eqnarray}
    D_{\kp, k}(\bZ) = \log \left(\frac{\pi_{\kp}}{\pi_{k}} \right) - \frac{1}{2} \log \left( \frac{\det(\Sigma_{\kp})}{\det(\Sk)} \right) - \frac{1}{2} \left(\bZ - \bmu_{\kp} \right)^{\top} \Sigma_{\kp}^{-1} \left(\bZ - \bmu_{\kp} \right) \notag 
    \\ + ~\frac{1}{2} \left(\bZ - \bmu_{k} \right)^{\top} \Sk^{-1} \left(\bZ - \bmu_{k} \right), \label{eqn:PDF} 
\end{eqnarray}
for $\kp, k \in \{1, \ldots, J\}$. The misclassification probability  corresponding to any generic classifier $\delta$ is defined as
\begin{equation}
\Delta = \sum_{k = 1}^{J} \pi_{k} \Prob \left(\delta(\bZ) \neq k \mid \bZ \in P_{k} \right). 
\label{eqn:MP_population}
\end{equation} 
We denote the misclassification probability corresponding to the Bayes classifier associated with QDA as $\Delta^{\mathrm{QDA}}$.

In practice, the population parameters $\pi_{k}$, $\bmu_{k}$ and $\Sk$ are unknown and estimated using the training samples. 
Typically, we are provided with $n = \sum_{k} n_{k}$ labeled training samples, denoted by $(\bX_{i}, y_{i})$ for $i= 1, \ldots, n$. Without loss of generality, we assume that the sample indices are ordered population-wise so that $\{\bX_{N_{k-1} + 1}, \ldots, \bX_{N_{k}} \}$ is a random sample of size $n_{k}$ from population $P_{k}$, where $N_{k}= \sum_{l = 1}^{k} n_{l}$ for $k = 1, \ldots, J$ and $N_{0} = 0$.
Based on the training sample, one estimates the parameters $\pi_{k}$, $\bmu_{k}$ and $\Sk$ as follows:
\begin{eqnarray}
    \widehat{\pi}_{k} = \frac{n_{k}}{n}, \quad \widehat{\bmu}_{k} =\frac{1}{n_{k}} \sum_{i=N_{k-1} + 1}^{N_{k}} \bX_{i}, \quad \text{and} \quad \widehat{\Sigma}_{k} = \frac{1}{n_{k} - 1} \sum_{i=N_{k-1} + 1}^{N_{k}} \left(\bX_{i} - \widehat{\bmu}_{k} \right)\left(\bX_{i} - \widehat{\bmu}_{k} \right)^{\top}
    \label{eqn:estimates}
\end{eqnarray} 
for $k = 1, \ldots, J$.

Fix $k_{0} \in \{ 1, \ldots, J\}$. The estimated discriminant function $\widehat{D}_{\kp, k}$ can be obtained by replacing the parameters by their estimates in $D_{\kp, k}$ as given by equation (\ref{eqn:PDF}). The estimated QDA classifier is defined as
\begin{equation} \delta^{\mathrm{QDA}}_n(\bZ) = 
    k_{0}  ~~ \text{if} ~  \widehat{D}_{k_{0}, k }(\bZ) > 0 ~~ \text{for all} ~~  k \in \{ 1, \ldots, J\} \setminus \{k_{0}\}. \label{eqn:Classif_Est} 
\end{equation}

When $p >n_{\max} \coloneqq \vee_{k} n_{k}$, the sample variance-covariance matrices $\widehat{\Sigma}_{k}$ for $k = 1, \ldots, J$ are not invertible. 
Use of generalized inverses may lead to sub-optimal results (see, e.g., \cite{bickel2004some}). Even if $p < n_{\min} \coloneqq \wedge_{k} n_{k}$, the computational complexity of QDA is of the order $O\left(n_{\max} p ^{2} \vee p^{3}\right)$ which is prohibitive for large $p$. One popular way to improve accuracy in the high-dimensional scenario (i.e., when $p \gg n)$ is to employ a dimension reduction technique prior to performing QDA. A convenient and computationally inexpensive approach towards dimension reduction is through the random projection ensemble (RPE), which we describe in the next subsection. 

\subsection{Random Projection Ensemble on QDA } \label{sec:RPE}

Let $R = ((r_{ij}))$ be a random matrix of dimension $d \times p$, where each component $r_{ij}$ is independent and identically distributed (i.i.d.) from $G$ for $i = 1, \ldots, d$, $j = 1, \ldots, p$ and $d \ll p$. Given $R$, the projected random variable $R \bZ$ belongs to one of the $J$ populations in the projected space (say, $P_{k}^{R}$) with mean $R\bmu_k$ and covariance matrix $R\Sk R^{\top}$ for $k = 1, \ldots, J$. If $\bZ$ is closer to the population $P_k$ in the $p$-dimensional ambient space, then one expects $R\bZ$ to be close to the population $P_k^{R}$ in $\mathbb{R}^{d}$ as well. 
The discriminant function on this projected subspace (say, $D_{\kp, k}^{R}$) 
can be obtained by replacing $\bZ$, $\bmu_{k}$ and $\Sk$ with their projected counterparts, $R\bZ$, $R\bmu_{k}$ and $R\Sk R^{\top}$, in the expression of $D_{\kp, k}$ in equation (\ref{eqn:PDF}), which is as follows:   
\begin{eqnarray}
    D_{\kp, k}^{R}(\bZ) &=& \log \left(\displaystyle\frac{\pi_{\kp}}{\pi_{k}} \right)  - \frac{1}{2} \log \left( \displaystyle\frac{\det \left(R \Sigma_{\kp} R^{\top} \right)}{\det(R \Sk R^{\top})} \right) - \frac{1}{2} \left(\bZ - \bmu_{\kp} \right)^{\top} \Psi_{\kp} \left(\bZ - \bmu_{\kp} \right) \notag \\
    && \hspace{2.5 in} + \frac{1}{2} \left(\bZ - \bmu_{k} \right)^{\top} \Psi_{k} \left(\bZ - \bmu_{k} \right) , \label{eqn:RPDF} 
\end{eqnarray}
where $\Psi_{k} = R^{\top} \left( R \Sk R^{\top} \right)^{-1} R$ for $k =1, \ldots, J$. Based on $D_{\kp, k}^{R}$, one may obtain a classifier (say, $\delta^{R\text{-}\mathrm{QDA}}$) which is similar to $\delta^{\mathrm{QDA}}$ in (\ref{eqn:Classif2}) with $D_{\kp, k}$ replaced by~$D_{\kp, k}^{R}$ for $k \neq \kp$.

The accuracy of the projected classifier $\delta^{R\text{-}\mathrm{QDA}}$ clearly depends on the random hyperplane induced by the particular choice of $R$. 
If the projection subspace retains the essential discriminative information, then one expects the classification accuracy of $\delta^{R\text{-}\mathrm{QDA}}$ to match with that of the original classifier $\delta^{\mathrm{QDA}}$. For instance, in a two-class location problem, let the location parameters be $\bmu_{1} = (\sqrt{p}, {\bf 0}_{p - 1}^{\top})^{\top}$ and $\bmu_{2} = {\bf 0}_{p}$. 
Consider a $d\times p$ random matrix $R$ with $d=1$, whose components are i.i.d. from the sparse three-point distribution STP\{-1, 0, 1\} (described later in (\ref{eqn:STP})). In such a case, the projected space can capture the relevant classification information if the first component of $R$, say $r_{1}$, satisfies $r_{1} \in \{\pm 1\}$, which happens with probability $p^{-1/2}$. Conversely, there are alternative choices of $R$ (with $r_1 = 0$) that would yield a classifier which may be lacking in its discriminative power.

To limit subjectivity regarding the choice of projection matrix $R$, we adopt the random projection ensemble (RPE) approach where an aggregation over the outcomes of multiple independent projection matrices is considered. Towards that, consider $B$ independent $d\times p$ random matrices $R_{1}, \ldots, R_{B}$. For each random matrix $R_{b}$, let the discriminant function in the projected plane be $D^{b}_{k^{\prime},k} \coloneqq D^{R_b}_{k^{\prime},k}$ for $b = 1, \ldots, B$. 
Fix $k_{0} \in \{1, \ldots, J\}$.
The RPE-QDA discriminant function is defined as
\begin{eqnarray}
    D^{\text{RPE}}_{k^{\prime}, k}(\bZ) = \frac{1}{B} \sum_{b=1}^{B} D^{b}
_{k^{\prime}, k}(\bZ), \label{eqn:DRPE}
\end{eqnarray}
while the RPE-QDA classifier is as follows:
\begin{eqnarray}
\delta^{\mathrm{RPE\text{-}QDA}}(\bZ) =  k_{0}  ~~ \text{if} ~  D_{k_{0}, k }^{\mathrm{RPE}}(\bZ) > 0 ~~ \text{for all} ~~  k \in \{ 1, \ldots, J\} \setminus \{k_{0}\}.   \label{eqn:RPClassif}
\end{eqnarray}
Again, the unknown population parameters $(\pi_{k}, \bmu_{k}, \Sk)$ are estimated by their sample counterparts $(\widehat{\pi}_{k}, \widehat{\bmu}_{k}, \widehat{\Sigma}_k)$ as defined in (\ref{eqn:estimates}) for $k= 1, \ldots, J$.
Define $\widehat{D}^{\text{RPE}}_{k^{\prime}, k}$ as the estimated RPE-QDA discriminant function obtained by replacing $(\pi_{k}, \bmu_{k}, \Sk)$ with $(\widehat{\pi}_{k}, \widehat{\bmu}_{k}, \widehat{\Sigma}_k)$ in equation (\ref{eqn:DRPE}). 
The corresponding classifier (say, ${\delta}^{\text{RPE-QDA}}_{n}$) is called the sample RPE-QDA classifier.

\paragraph{Advantages of the RPE-QDA Approach.} 
%in Ultrahigh dimensions.}
In ultrahigh-dimensional data sets, discriminative information is often concentrated within a low-dimensional subspace.
This information may not be captured by selected summary measures, such as the mean, trace or highest eigenvalue of the covariance matrix. 
Dimension reduction via random projections provides an attractive and effective approach for handling such data sets. 
Since the dimension of the projected space \(d\) is typically much smaller than the minimum sample size \(n_{\min}\), applying sample QDA to each randomly projected subspace retains the optimality properties of QDA in the small \(p\), large \(n\) setting. Thus, employing QDA across multiple random projections, followed by pooling of information is expected to yield an overall improved classification~accuracy.

Computational complexity of the sample RPE-QDA classifier is of the order $O \left(B \left\{d p n_{\max} + d^{3} \right\}\right)$, 
where $dpn_{\max}$ corresponds to the projection step (i.e., multiplication of $R$ with $\bX_i$ for $i = 1, \ldots, n$) and $d^3$ accounts for the inversion of the projected covariance matrix.
The multiplicity factor $B$ is due to the ensemble step. Noticeably, the $B$ replications are \emph{embarrassingly parallel}. Therefore, the computational time is quite low in practice. 
Typically, a moderately large choice of $B$ suffices, so that $p \gg \max\{ B, n_{\max} \}$ and $\min\{ B, n_{\min} \} \gg d$  (see Section \ref{sec:realdata-analysis}), resulting in a total complexity of $o(p^{2})$, which is a substantial reduction compared to the usual computational cost of the classical QDA.
%, which is of the order $O\left(n_{\max}p^{2} \vee p^{3}\right)$.  

\vskip5pt
\noindent{\bf Choice of Random Matrix.} 
The selection of the random matrix is critical to the effectiveness of the RPE-QDA method. 
%{\color{red}As described earlier, consider a random matrix $R = \left( (r_{ij})\right)$, where $r_{ij}$ (for $i = 1, \ldots, d$, $j = 1, \ldots, p$) is i.i.d. from the distribution $G$.} 
As described earlier, the selected random matrix $R = \left( (r_{ij})\right)$ has all its components generated %independently and identically from the distribution 
i.i.d. from $G$. 
Conventional choices of $G$ include the standard Gaussian distribution, and symmetric two or three point distributions supported on $\{-1, 1\}$ or $\{-1, 0, 1\}$ with equal probabilities (see, e.g., \cite{achlioptas2001database}). 
%These matrices satisfy the Johnson-Lindenstrauss (JL) lemma {\color{red}(see ???)}.

\cite{li2006very} proposed an extreme sparse choice of the three point distribution as follows  
\begin{equation}
    r_{ij} = \begin{cases} \pm 1 \qquad \text{ with probability }\quad (2\sqrt{p})^{-1}, \\
    ~~0 \qquad ~ \text{with probability }\quad 1 - (\sqrt{p})^{-1}.
    \end{cases} \label{eqn:STP}
\end{equation}
We shall denote this sparse three point distribution as $\mathrm{STP}\{-1, 0, 1\}$. This choice of the random matrix also satisfies the JL lemma (see \cite{dasgupta2003elementary}) and is efficient in ultrahigh-dimensional sparse situations.
In this paper, we use the standard Gaussian distribution with unit variance as the choice of $G$ for our theoretical investigations. We additionally consider the choice $G \equiv \mathrm{STP}\{-1, 0, 1\}$ in our numerical work.

\section{Theoretical Results}\label{sec:theory}

Projection of ultrahigh-dimensional data to random low-dimensional subspaces may clearly lead to loss in discriminative power of a classifier. It is important to establish a set of sufficient conditions under which this loss is minimal, and also to quantify the amount of loss. Towards that, we conduct a comprehensive theoretical analysis investigating the asymptotic properties of the proposed RPE-QDA classifier in this section.

We first investigate the \emph{perfect classification} property (in an asymptotic sense) of the Bayes classifier based on QDA, followed by that of the proposed RPE-QDA classifier and its sample analog. 
Under appropriate structural assumptions on the mean vectors and covariance matrices, we first establish a connection between the discriminant functions of the aforementioned classifiers and the Kullback–Leibler (KL) divergence between the competing populations (see, e.g., \cite{KL1951}). Leveraging on this relationship, we prove that the misclassification probabilities of these classifiers converge to \emph{zero} as $\min\{B, n, p\} \to \infty$.

\subsection{Setup and Assumptions}

The perfect classification property will be examined under the following set of conditions:
\begin{enumerate}\label{assumptions}    
\itemsep-0.25em
    \item[(A1)] \hypertarget{A1} {{\bf Mean structure:}} $\max_{k \in \{1, \ldots, J\}} \| \bmu_{k} \|^{2} = O\left(p\right)$ as $p \to \infty$. 

    \item[(A2)] \hypertarget{A2} {{\bf Covariance structure:}} 
    \begin{enumerate}
    \item $\max_{k\in \{1, \ldots, J\}} \LM(\Sk) = O\left(p^{\alpha}\right)$ for some $0 \leq \alpha<1$.

    \item Let $\xi \in (0,1)$. For each $k  =1, \ldots, J$, at most $m_{k,p} \lesssim O\left(p^{\xi}\right)$ eigenvalues of $\Sk$ are different from $\Lm (\Sk) = \gamma_{k}~ (\neq 0)$ and $\log p \left\{\lambda_{1}(\Sk) + \cdots + \lambda_{m_{k,p}} (\Sk) \right\}/ p \to 0$  as $p \to \infty$.

    \end{enumerate}

% The restriction in part (b) implies that the trace of $\Sk$ satisfies $\tr(\Sk) = p \gamma_{k} (1 + o (1))$ for $k= 1, \ldots, J$.

    \item[(A3)] \hypertarget{A3} {{\bf Separability:}} Let $\text{KL}_{k,k'}$ denote the KL divergence of $N_p(\bm{\mu}_k,\Sigma_k)$ with respect to $N_p(\bmu_{k'},\Sigma_{k'})$, i.e.,
    $$ 2 \mathrm{KL}_{k,k'} = \tr \left(\Sigma_{k}^{-1} \Sigma_{k'} \right) + (\bmu_{k}-\bmu_{k'})^{\top} \Sigma_{k}^{-1} (\bmu_{k} -\bmu_{k'}) -p + \log \det\left(\Sigma_{k} \Sigma_{k'}^{-1} \right) .$$

    Assume that
    $\liminf_{p} p^{-1} \min_{k,k'}\{\text{KL}_{k,k'}\} \geq \nu_{0}$ for some $\nu_{0}>0$, with $k \neq \kp \in \{1, \ldots, J\}$.
\end{enumerate} 

\noindent
 
Assumption (\At) imposes a reasonable restriction on the growth of means as $p \to \infty$. In particular, if all the components of the means are bounded, then (\At) holds. 
The conditions in (\Ao) encompass a broad class of covariance matrices that includes strongly spiked eigenstructure (see, e.g., \cite{aoshima2019distance}), or eigenstructures of the identity type matrices. We can broadly characterize the class of matrices satisfying (\Ao) by $\Sigma = \gamma I +  \Gamma \Gamma^{\top}$, where $\Gamma$ is a $p\times r$ matrix with highest singular value of order $p^{\alpha}$ ($0 \leq \alpha<1$), $r\ll p$, and $\gamma>0$ is smaller than the lowerst non-zero eigenvalue of $\Gamma \Gamma^{\top}$.  
Finally, any classification method requires a minimum level of separation between the competing populations. Assumption (\Ath) essentially states that the KL divergence (a measure of separations between the two distributions) should be at least of the order of $p$.

The aforementioned assumptions are all valid under several situations. 
In particular, let us consider the following two class examples.
\begin{example}\label{ex1}
Consider a problem with $P_1 \equiv N_p(\bm{1},\Sa)$ and $P_2 \equiv N_p(\bm{2},\Sigma_2)$. Define $\Sigma_{k} = \mathrm{Diag}\big(\Gamma_{p_{k}}, \gamma_{k} I_{p - p_{k}} \big) $,
% $$\Sigma_k = \begin{bmatrix}
%     \Gamma_{p_k} & O  \\
%     O & \gamma_{k} I_{p - p_{k}} \label{eqn:covmat_example} \\ 
% \end{bmatrix} ,$$
\hypertarget{par:example}{}
where $\Gamma_{p_{k}} = (1- \rho_{k}) {I_{p_{k}}} + \rho_{k} \bm{1} \bm{1}^{\top}$, $\max \{p_{1}, p_{2} \} = p^{\alpha}$ with $0 <\alpha <1$, $0\leq \rho_{k} <1$ and $0<\gamma_{k} \leq (1 - \rho_{k})$ for $k = 1,2$.
Now, (\At) holds as $\max_{k\in \{1,2\}}\|\bmu_{k}\|^{2} \leq 4p$, 
while (\Ao) holds as the highest eigenvalue of $\Sk$ is of order $p_{k}$ and $\tr(\Sk) \log p = o(p)$ for $k =1,2$.  
For (\Ath), observe that
\begin{eqnarray}
2 \mathrm{KL}_{1,2} 
\geq  (\bmu_{1}-\bmu_{2})^{\top} \Sigma_{1}^{-1} (\bmu_{2} -\bmu_{2}) 
 \geq  {\bf 1}^{\prime} \Gamma_{p_{1}}^{-1} {\bf 1} + \frac{1}{\gamma_{1}} (p - p_{1}) 
 \geq  \frac{1}{\gamma_{1}} ( p - p_{1}) \geq \frac{p}{2 \gamma_{1}}  \notag
\end{eqnarray}
as $p_{1}/p \to 0$. 
Similarly, $2 \mathrm{KL}_{1,2}  \gtrsim  p/(2 \gamma_{2})$. Thus, (\Ath) holds with $\nu_{0} = (2 \gamma_{1} \vee \gamma_{2})^{-1}$.
% Thus, assumption (\Ath) holds with $ \nu_{0} = (\nu_{12} \wedge \nu_{21})/2$. 
%}
\end{example}
\noindent One may extend Example \ref{ex1} by taking $\lfloor p^{\beta}\rfloor$ (with $\beta \leq (1 - \alpha)$) many block equi-correlation matrices. Assumptions (\At)-(\Ath) can be verified by following similar calculations as above.

\begin{example}\label{ex2}
Let $P_1 \equiv N_p(\bm{0},\Sigma)$ and $P_2 \equiv N_p(\bm{0}, c\Sigma)$ with $\Sigma = I_{p} + P D P^{\top}$, $D = \mathrm{Diag}(\gamma_{1}, \ldots, \gamma_{r})$ and $ c >0$, $c\neq 1$. Here, $P$ is a $p \times r$ matrix satisfying $P^{\top} P = I_{r}$, $r = o(p)$ and $1 \leq \gamma_{j}< M$ (with some $M>1$) for $j =1, \ldots,r$. Clearly, assumptions (\At) and (\Ao) are satisfied here. Observe that
$ 2 \mathrm{KL}_{1,2} = (c - \log c - 1)p$ and $2 \mathrm{KL}_{2,1} = (c^{-1} - \log c^{-1} - 1)p$.
Thus, (\Ath) holds from the fact that the function $h(x) = x - \log x -1$ is positive for all $x \in (0,\infty)\setminus\{ 1\} $. 
\end{example}

\begin{example}\label{ex_extra}
    Consider a two class problem with $P_1 \equiv N_p(\bm{0},\Sa)$ and $P_2 \equiv N_p(\bm{0},\Sigma_2)$. Define $\Sigma_{1} = \mathrm{Diag}\big(\Gamma_{\lceil p/2 \rceil},  I_{\lfloor p/2 \rfloor} \big) $ and $\Sigma_{2} = \mathrm{Diag}\big(  I_{\lfloor p/2 \rfloor}, \Gamma_{\lceil p/2 \rceil} \big) $,
where $\Gamma_{r} = P D P^{\top}$ with $D = \mathrm{Diag}(p^{\alpha}, 2, \ldots, 2)$ for $0<\alpha <1$ and $PP^{\top} = P^{\top}P = I_{r}$.
Clearly, (\At) and (\Ao) hold as the highest eigenvalue of both $\Sa$ and $\Sigma_{2}$ is $p^{\alpha}$. Further,
$$ 2 \mathrm{KL}_{1,2} = \tr \left( \Sa^{-1} \Sigma_{2} \right) - p = \tr \left( D_{\lceil p/2 \rceil}^{-1}\right) + \tr \left( D_{\lfloor p/2 \rfloor} \right) - p = p^{-\alpha} + \frac{1}{2} \left\lceil \frac{p}{2} \right\rceil  +p^{\alpha}  + 2 \left\lfloor \frac{p}{2} \right\rfloor - p.$$
Similarly, $\mathrm{KL}_{2,1}\sim p/4$. Thus, (\Ath) holds with $\nu_{0} = 0.25$.
\end{example}

\begin{example}\label{ex3}
Consider the problem of classifying $P_1 \equiv N_p(\bmu_{1},\Sigma)$ and $P_2 \equiv N_p(\bmu_{2},\Sigma)$, where the common variance $\Sigma$ is as in Example \ref{ex2}.
In this case, the KL divergence will simply be proportional to the Mahalanobis distance between $P_1$ and $P_2$, i.e., $2\mathrm{KL}_{1,2} = 2\mathrm{KL}_{2,1} = (\bmu_2-\bmu_1)^{\top} \Sigma^{-1} (\bmu_2-\bmu_1) \sim \lVert \bm{\mu}_1 - \bm{\mu}_2 \rVert^2$.  
Assumptions (\At)-(\Ath) hold for any $\bmu_{1}$ and $\bmu_{2}$ satisfying $\lVert \bm{\mu}_1 - \bm{\mu}_2 \rVert^2 \gtrsim p$ and $\|\bmu_{k}\|^{2} = O\left(p\right)$ for $k=1,2$. 

In particular, $\bmu_{1} = (\sqrt{p}, {\bf 0}_{p-1}^{\top})^{\top}$ and $\bmu_{2} = {\bf 0}_{p}$ provides a choice where sufficient difference in only one index of the mean vector satisfies the requirement of (\Ath), with no scale difference.
%in the variance components. 
\end{example}

Under the set of sufficient conditions (\At)-(\Ath), we are now ready to provide a framework under which the Bayes risk of the classical QDA (see equation (\ref{eqn:Classif2})) achieves
{\it perfect classification} in a limiting sense as $p \to \infty$.

%{\color{red}
\begin{theorem}\label{thm:0}
Let $\bZ$ be a test sample from one of the $J$ populations $P_{k}$ with probability distribution $N_p ( \bmu_{k}, \Sk )$ and $\Prob(\bZ \in P_{k} )= \pi_{k} \in (0,1)$ for $k \in \{1,\ldots, J\}$. Consider the classifier $\delta^{\mathrm{QDA}}(\bZ)$ defined in \eqref{eqn:Classif2} and the corresponding discriminant function $D_{k,\kp}$ defined in (\ref{eqn:PDF}). Under assumptions (\At)-(\Ao), the following statements hold:
\vspace{-.075 in}
\begin{enumerate}[(a)]
\itemsep-0.25em
\item If $\bZ \in P_{k}$, then $p^{-1}\left\{ D_{k,\kp} (\bZ) -  \mathrm{KL}_{k,\kp} \right\} \xrightarrow{P} 0$ as $p \to \infty$ for all $\kp \neq k$.
\item Additionally, under (\Ath), the misclassification probability of the classical QDA classifier $\Delta^{\mathrm{QDA}}$ satisfies 
$\lim_{p \to \infty} \Delta^{\mathrm{QDA}} = 0.$
\end{enumerate}
\end{theorem}

Theorem \ref{thm:0} states that under assumptions (\At)–(\Ao), if the test point $\bZ \in P_k$, then the behavior of the QDA discriminant function comparing $P_k$ and $P_{\kp}$ closely aligns with the KL divergence from $P_k$ to $P_{\kp}$ (after appropriately scaling with $p$). This alignment results from the fact that $D_{k,\kp}$ essentially captures the difference in log-likelihoods, as the influence of prior probabilities diminishes when $p \to \infty$.
%, i.e., when sufficient information is available. 
Moreover, the quantity $\mathrm{KL}_{k,\kp}$ corresponds to the expectation of the difference in log-likelihoods under $P_k$. Therefore, when $\bZ \in P_k$, the result follows naturally as a consequence of the law of large numbers.
Assumption (\Ath) guarantees sufficient KL divergence, which ensures that the misclassification probability of the classical QDA classifier converges to zero. This also indicates that the scaling factor $p$ for the KL divergence in (\Ath) is minimally required for the classical QDA to achieve {\it perfect~classification}.

\subsection{Perfect Classification of RPE-QDA and Sample RPE-QDA} \label{PC_RPE-QDA}

We now investigate the perfect classification property of the RPE-QDA classifier as the number of random projection matrices $B \to \infty$ (for each fixed $p$), followed by $p \to \infty$.

The RPE-QDA classifier, defined in equation (\ref{eqn:RPClassif}), is a stochastic classifier owing to the randomness of the matrices $R_{1}, \ldots, R_{B}$ given a test sample $\bZ$. The misclassification probability corresponding to RPE-QDA is denoted by $\Delta^{\text{RPE-QDA}}$.
In Theorem \ref{thm:1}, we show that $\Delta^{\text{RPE-QDA}}$ also converges to \emph{zero} under a Gaussianity assumption on the random matrices, when the number of independent random matrices $B \to \infty$ and $p \to \infty$. 
At this point, it is important to note the order in which the two limits are taken. 
For a given dimension $p$, we consider $B$ independent random matrices of order  $p \times d$, where $d \equiv d(p)$ is a function of~$p$. 
Therefore, we first consider the limit with respect to $B$ for each fixed $(p, d)$, and then let $p \to \infty$ to establish \textit{perfect classification} of $\delta^{\text{RPE-QDA}}$. 

%{\color{red} 
\begin{theorem}\label{thm:1}
Let $\bZ$ be a test sample from one of the $J$ populations $P_{k}$ with probability distribution $N_p ( \bmu_{k}, \Sk)$ and $\Prob(\bZ \in P_{k} )= \pi_{k} \in (0,1)$ for $k \in \{1,\ldots, J\}$. 
Consider the classifier $\delta^{\mathrm{RPE\text{-}QDA}}(\bZ)$ defined in equation (\ref{eqn:RPClassif}) with $G \equiv N(0,1)$, and the corresponding discriminant function $D_{k, \kp }^{\mathrm{RPE}}(\bZ)$. Under assumptions (\At)-(\Ao), the following hold:
\vspace{-.075 in}
\begin{enumerate}[(a)]
\itemsep-0.25em
\item Define $D_{k,\kp}^{\mathrm{RP}}(\bZ):= \Exp_{R} \big [ D_{k,\kp}^{\mathrm{RPE}}(\bZ) \big ]$.
%where the expectation taken is with respect to the random matrices, given $\bZ$. 
For each fixed $p$, we have
$$ D_{k, \kp }^{\mathrm{RPE}}(\bZ) \xrightarrow{a.s.} D_{k,\kp}^{\mathrm{RP}}(\bZ) \mbox{ as } B \to \infty.$$

\item If $\bZ \in P_{k}$ and $d = O\left(\log p\right)$ with $d \to \infty$ as $p\to \infty$, then 
$$ \left[ d^{-1} D_{k,\kp}^{\mathrm{RP}}(\bZ) - p^{-1} \mathrm{KL}_{k,\kp} \{ 1 + o(1)\} \right] \geq 0$$
with probability tending to one as $p \to \infty$.
 
\item Additionally, under (\Ath), the misclassification probability of the RPE-QDA classifier $\Delta^{\mathrm{RPE\text{-}QDA}}$ satisfies 
$$ \lim_{p \to \infty} \lim_{B \to \infty} \Delta^{\mathrm{RPE\text{-}QDA}} = 0 .$$
\end{enumerate}
\end{theorem}
%}

\noindent
In Theorem \ref{thm:1}, the projected dimension $d$ is required to increase with $p$ as well.
The growth of $p$ signifies an increasing accumulation of information, and to accommodate this, $d$ must grow, albeit at a much slower rate.

\begin{remark}\label{rm:1}
A comparison between parts (a) of Theorem \ref{thm:0} and (b) of Theorem \ref{thm:1} clearly highlights the loss in information induced by the use of RPE. While the QDA discriminant function requires a scaling factor of order $p$ (consistent with KL divergence), the RP-QDA discriminant involves a significantly smaller scaling factor of $d$. This reduction reflects a diminished discriminative power of the RP-QDA classifier due to this dimension reduction. Moreover, as expected, increasing the value of the reduced dimension $d$ leads to an improvement in the discriminative performance.
\end{remark}

\paragraph{Sample RPE-QDA classifier.} Next, we turn our attention to the sample version of the RPE-QDA classifier, ${\delta}^{\text{RPE-QDA}}_{n}$, proposed in Section \ref{sec:RPE}. This is obtained by estimating the population parameters  $(\pi_{k}, \bmu_{k}, \Sk)$ for $k = 1, \ldots, J$ with their sample analogs (defined in (\ref{eqn:estimates})) in the classifier $\delta^{\text{RPE-QDA}}$.
The corresponding misclassification probability $\Delta^{\text{RPE-QDA}}_{n}$ is obtained by using the empirical RPE-QDA classifier in equation \eqref{eqn:MP_population}. 
It is important to observe that $\Delta^{\text{RPE-QDA}}_{n}$ includes randomness from the random projections as well as that of the training samples.
Theorem \ref{thm:2} below provides a set of sufficient framework under which the classifier ${\delta}^{\text{RPE-QDA}}_{n}$ yields \textit{perfect classification} in a limiting sense
as $B \to \infty$ and $\min\{n, p\} \to~\infty$.

%{\color{red} 
For each set of training samples, $B$ independent random matrices are generated to obtain the RPE-QDA classifier. As is typical in the ultrahigh-dimensional asymptotic regime, we consider the dimension $p$ to be a function of the sample size $n$ (although no specific restriction has been imposed on the growth rate of $p$ relative to $n$). Therefore, the limit is first taken with respect to $B$, followed by the limit on $(n,p)$.
%}
% {\color{red} For each set of training samples, the estimates $\{ (\widehat{\bmu}_{k}, \widehat{\Sigma}_k); k = 1, \ldots, J\}$ are computed first. Subsequently, several random matrices are generated to construct the classifier given $p$. 
% As such, the limit is taken first with respect to the number of random matrices $B$, followed by the limit as $(n,p)$ grows. In this context, we impose no specific restriction on the growth rate of $n$ relative to $p$. In fact, $n$ may increase at a much slower rate than $p$.}

\begin{itemize}
\item [(A4)] \hypertarget{A4}{{\bf Rates of increase of $\{n_1, \ldots,n_{J},d\}$:}} Assume that $n_k/n_{\kp} \to \beta^{k,\kp} \in (0, \infty)$ as $n \to \infty$ for all $k,k' \in \{1, \ldots, J\}$. Further,
the reduced dimension $d$ satisfies $d = o\left( n_{\min}\right)$ and $d = O\left(\log p\right)$ with $d\rightarrow \infty$ as $p\to \infty$.
\end{itemize}

\noindent
The first part of assumption (\Af) is a balanced design condition, which restricts uniform dominance of one group over the others. 
There is no explicit assumption on the rate of growth of $n$ compared to $p$.
The reduced dimension $d$ is allowed to grow at a slower rate than $n_{\min}$. 
This in turn ensures that $R\widehat{\Sigma}_k R^{\top}$ is invertible with probability $1$ for all $k \in \{1, \ldots, J\}$.
In a typical ultrahigh-dimensional scenario, one assumes $p \sim \exp\{o(n)\}$. The condition $d \sim O(\log p)$ is equivalent to $d \sim o(n)$, which is compatible with (\Af) in such situations. 

\begin{theorem}\label{thm:2}
Let $\bZ$ be a test sample from one of the $J$ populations $P_{k}$ with probability distribution $N_p ( \bmu_{k}, \Sk)$ and $\Prob(\bZ \in P_{k} )= \pi_{k} \in (0,1)$ for $k \in \{1,\ldots, J\}$. 
Consider the sample classifier $\delta_n^{\mathrm{RPE\text{-}QDA}}(\bZ)$ with $G \equiv N(0,1)$. Under assumptions (\At)-(\Af), the misclassification probability of the sample RPE-QDA classifier $\Delta_n^{\mathrm{RPE\text{-}QDA}}$ satisfies 
$$ \lim_{\min\{n, p\} \to \infty} \lim_{B \to \infty} \Delta_n^{\mathrm{RPE\text{-}QDA}} = 0.$$
\end{theorem}

\begin{remark}
Parts (a) and (b) of Theorem \ref{thm:1} can similarly be established for the sample version of the RPE-QDA classifier, i.e., these results hold even if we replace $D_{k,\kp}^{\mathrm{RP}}$ and $D_{k,\kp}^{\mathrm{RPE}}$ by their sample counterparts. Therefore, $\widehat{D}^{\mathrm{RP}}_{k,\kp}$ entails no additional loss of information compared to its population analog. The same scaling factor $d$ obtained in $D^{\mathrm{RP}}_{k,\kp}$ (see part (b) of Theorem \ref{thm:1}) remains appropriate for $\widehat{D}^{\mathrm{RP}}_{k,\kp}$ too. 
\end{remark}

\noindent
As a common thread among Theorems \ref{thm:0}-\ref{thm:2} we consider the limit as $p\to \infty$.  This is because of the implicit assumption that separability among the populations becomes significant as the dimension $p$ grows (see also assumption (\Ath)). To control the variability incurred from using different random projections in RPE-QDA, we require $B \to \infty$. To further control the fluctuations of the training sample, a balanced design condition along with $n \to \infty$ is necessary (see (\Af)) for consistency of the sample RPE-QDA classifier. 
We conclude this section with a discussion and additional insights into the RPE-QDA approach by comparing it with some existing methods in the literature. 

\vskip5pt
\noindent{\bf A comparison with existing literature.} 
As mentioned in the Introduction, several approaches to high-dimensional QDA rely on specific summary measures. For instance, the method proposed by \cite{AoYa2014} (say, AoYa) captures class separation only through differences in means and the traces of covariance matrices. 
\cite{wu2019quadratic} proposed the \emph{ppQDA} approach, which additionally utilizes the sum of off-diagonal elements to enhance discrimination. 
From a theoretical standpoint, both methods establish convergence of the Bayes risk to zero. However, a significant separation among the related summary measures is necessary for these methods to yield desirable results.   
In a two-class problem,
a key assumption for consistency of AoYa is $ \tr(\Sigma_1^2) \vee \tr(\Sigma_2^2) / \{n_k \lVert\bm{\mu}_1 - \bm{\mu}_2 \rVert^2\} \to 0$ for $k \in \{1,2\}$, which explicitly excludes cases with no mean separation.
For ppQDA, a key assumption is $|(a_1 - r_1) - (a_2 - r_2)| > \delta_0 > 0$,
where ${a}_k = p^{-1}\tr({\Sk})$ and ${r}_k = \{p(p-1)\}^{-1} \left\{\bm{1}^{\top}{\Sk}\bm{1} - \tr({\Sk}) \right\}$ for $k \in \{1,2\}$.
This condition requires sufficient discriminative information in the averages of the diagonal and off-diagonal elements of the covariance matrices.
However, it is possible for the population distributions to differ significantly even when these summary measures remain identical.
For instance, in Example \ref{ex_extra}, we have $\bmu_{1} = \bmu_{2} = {\bf 0}$, $a_{1} = a_{2}$ and $r_{1} = r_{2}$, although the competing populations differ significantly, as evidenced by a KL divergence of order $0.25 p$.

In contrast, RPE-QDA is able to distinguish between populations that differ significantly in any form, not necessarily limited to those captured via selective summary measures. 
The relation between RPE-QDA discriminant function and KL divergence, established in Theorem \ref{thm:1} (also see Remark \ref{rm:1}), demonstrates its potential, while also quantifying the loss due to RPE.   
Recently, \cite{palias2023effect} derived an upper bound on the Bayes error for RP-QDA using a single standard normal random matrix. This bound depends on the summary measures such as $\lVert\bm{\mu}_1 - \bm{\mu}_2 \rVert^2$, $\LM(\Sigma_k)$, and $\tr(\Sigma_k)$ for $k=1,2$. When these quantities do not vary across populations, the bound becomes uninformative. For instance, in Example \ref{ex_extra}, the bound yields the trivial value 1, although RPE-QDA demonstrates high discriminative ability in such settings (see Scheme 4 in Section \ref{sec:simulation}). In this work, we have aimed to unveil the full potential of RPE-QDA through our theoretical results.

%To conclude, RPE-QDA is able to distinguish between populations that differ significantly in any form, not necessarily limited to those captured using selective summary measures.
Clearly, our theoretical results are only asymptotic in nature, and the finite sample performance of these classifiers is of significant practical interest and importance. To assess this further, we conduct some numerical experiments in the next two sections.

\section{Simulation Experiments}\label{sec:simulation}

In this section, we analyze four different simulation scenarios for the two class problem. 
For each setting, the data dimension $p$ ranges from  $500$ to $10000$ in the logarithmic scale (with step size $2$).
We set the train and test sample sizes as $100$ and $200$, respectively, for each class. 

\vskip5pt
\noindent {\it Competitors.} We compare the performance of RPE-QDA with the following methods: high-dimensional discriminant analysis (HDDA) proposed by \cite{bouveyron2007high}); IIS-SQDA classifier developed by \cite{fan2015innovated};  direct
approach for QDA (DA-QDA) proposed by \cite{jiang2018direct};
distance-based classifier developed by \cite{AoYa2014} (AoYa); QDA based on a pooling operation on covariance matrices (ppQDA) proposed by \cite{wu2019quadratic}; and bagging-based RPE method of\cite{cannings2017random}~(RPE-CS). 

\vskip5pt
\noindent {\it Proposed method.} We consider two variants of the proposed $\text{RPE-QDA}$ method to accommodate different choices of random matrices.  
The first variant, denoted by $\text{RPE-SN}$, has i.i.d. entries of a random matrix from the standard normal distribution, i.e., $G \equiv N(0,1)$. The second variant, termed as $\mathrm{RPE\text{-}STP}$, considers $G \equiv \mathrm{STP}\{-1, 0, 1\}$ (see (\ref{eqn:STP})).
R codes corresponding to RPE-SN and RPE-STP are available from \href{https://github.com/AnneshaDeb99/RPQDA.git}{this link}.

\vskip5pt
\noindent{\it Choice of tuning parameters:} 
Theoretical results suggest a choice of the intrinsic dimension $d$ of order $O(n_{\min} \wedge \ceil{\log p})$. However, to maintain a fixed choice across all simulation settings, we set $d = n_{\min} \wedge  \max_{p} \ceil{\log p} = 10$, where the maximum is taken over all the values of $p$ considered here.
%\vskip5pt
%\noindent{\it Choice of $B$:}  
From a theoretical perspective, a large choice of $B$ is ideal. Although, in practice, a moderately large choice suffices. Therefore, we set $B = 200$. 
% (we could drop as recommended by cannings part)

\vskip5pt
\noindent{\it Evaluation method.}
We compare the competing methods using the
empirical misclassification probabilities,
which is defined as  
${\widehat{\Delta}}_{n} = \hat{\pi}_{1} \hat{p}_{1} + \hat{\pi}_{2} \hat{p}_{2}, $
where $\hat{p}_{k}$ is the proportion of test samples from population $P_{k}$ that are incorrectly classified for $k=1,2$. 
The average and standard deviation (in parentheses) of empirical misclassification probability over $50$ replications for each of our simulation schemes are reported in Tables \ref{tab1}-\ref{tab4}~below. The minimum value of the average misclassification rate under each simulation setting is marked in bold. 
The Bayes risk is zero for all the simulation schemes considered in this section. To further assess the degree of separability, we report the value of $\mathrm{KL}/p \coloneqq \min{\{\mathrm{KL}_{1,2}, \mathrm{KL}_{2,1}\}}/p$ in each simulation scheme. 
%Here, $\mathrm{KL}_{i,j}$ denotes the KL-divergence of the $i$-th to the $j$-th population, $i,j\in \{1,2\}$. 
A higher value of $\mathrm{KL}/p$ indicates a larger separation between the probability distributions and is expected to correspond to a lower misclassification probability across all methods. 
%The minimum value in each column is marked in bold.

\vskip5pt
\noindent {\it Other specifications.} We implemented HDDA using the R package {\tt HDClassif}, and considered all the eigenstructures mentioned in \cite{bouveyron2007high} to select the one that provides the maximum reduction in empirical misclassification probability. For DA-QDA and IIS-SQDA, we used the R codes provided by the authors of \cite{jiang2018direct}. 
Additionally, for the IIS-SQDA, the R package {\tt glasso} is used to estimate the precision matrices under the sparsity assumption considered in \cite{fan2015innovated}. 
The method $\mathrm{RPE\text{-}CS}$ is implemented using the R package {\tt RPEnsemble}, with reduced dimension $d = 10$.
%For $\mathrm{RPE\text{-}CS}$ we have used the reduced dimension $d=10$. 
The number of random matrices at two stages are set to $B_1 = 500$ and $B_2 = 50$, as suggested by \cite{cannings2017random}. 

\vskip5pt

Our simulation settings are described below. 
The two candidate populations are denoted by $P_{1} \equiv N_p(\bmu_{1}, \Sa)$ and $P_{2} \equiv N_p(\bmu_{2}, \Sigma_{2})$, respectively. 
In the first simulation setting, we consider location as well as scale differences, whereas in the remaining settings, we consider only scale~differences. In each scheme, the prior probabilities for the two classes are taken to be equal, i.e., $\pi_1 = \pi_2 = 1/2$.

\vskip5pt
\noindent{\bf Scheme 1:} 
Consider
$\bmu_{1} = {\bf 0}$, $\bmu_{2} = \begin{bmatrix} {\bf 0} : {\bf 1}_{l}^{\top} : -{\bf 1}_{l}^{\top} \end{bmatrix}$ and  $\Sigma_k = \mathrm{Diag}\left( \Gamma_{p_{k(1)}},  \Gamma_{p_{k(2)}}, c_k \Omega_{p_{k(3)}}(\rho) \right)$ for $k = 1,2$, 
where $l = \lfloor p^{3/5}/2 \rfloor$, $ \Gamma_{t} = 0.5 I_{t} + 0.5 {\bf 1} {\bf 1}^{\top}$, $\Omega_t(\rho) = B_t \left((\rho^{|i-j|}) \right) B_t$ and $B_{t}= (1.5 + 1/t)^{1/2} I_{t}$.
%$B_t = \text{diag}({(1.5 + 1/t)}^{1/2},\ldots, {(1.5+t/t)}^{1/2})$.
We further set
$p_{1(1)} = \lfloor p^{2/3}\rfloor$, $p_{1(2)} = \lfloor p^{1/3}\rfloor$, $p_{1(3)} = p - p_{1(1)} - p_{1(2)}$, 
$p_{2(1)} = \lfloor p^{1/2}\rfloor$, $p_{2(2)} = \lfloor p^{1/2}\rfloor$, $p_{2(3)} = p - p_{2(1)} - p_{2(2)}$, 
$(c_1, c_2) = (1, 1.3)$ and $\rho = 0.7$.
This example has been analyzed by \cite{aoshima2019distance}.

\vskip5pt
\noindent{\bf Scheme 2:} Set $\bm{\mu}_1 = \bm{\mu}_2 = \bm{0}$.
Consider a block diagonal matrix $\Sa$ with the first $\lfloor p^{\beta_1} \rfloor$ diagonal blocks having an equi-correlation structure, with $\rho = 0.9$ and dimension $\lfloor p^{\alpha_1} \rfloor$, while the last block is an identity matrix. Define $\Sigma_2$ in a similar way but with a different set of parameters $\alpha_2$ and $\beta_2$. We take $\alpha_1 = 0.6$, $\beta_1 = 0.4$ and $\alpha_2 = 0.7$, $\beta_2 = 0.3$.

\vskip5pt
\noindent{\bf Scheme 3:} Consider $\bm{\mu}_1 = \bm{\mu}_2 = \bm{0}$, and $\Sigma_{1} = c \Sigma_{2}$.
The precision matrix of $\Sigma_{1}$ is set to $\Sigma_{1}^{-1} = \left( ( 0.9^{|i - j|}) \right)$,
for all $i, j \in \{1, \ldots, p\}$ and $c^{-1} = 1.3$.

\vskip5pt

\noindent{\bf Scheme 4:} Consider $\bm{\mu}_1 = \bm{\mu}_2 = \bm{0}$, 
$\Sa = \mathrm{Diag} \left(P\Lambda_{l}P^{\top}, I_{p - l} \right)$
and 
$\Sigma_{2} =  \mathrm{Diag} \left(I_{p - l}, P\Lambda_{l}P^{\top} \right)$,
where $l = \lfloor p^\beta \rfloor$, $\Lambda_{l} = \mathrm{diag}(p^\alpha, (p^\alpha -1)\ldots, (p^\alpha - l + 1))$ and $PP^{\top} = P^{\top}P = I_l$.
We set $\alpha = 0.6$ and $\beta = 0.5$. 

\begin{table}[ht]
\centering
    \renewcommand*{\arraystretch}{1.2}
    \begin{small}
\begin{tabular}{| c | c  c  c  c c c |}
  \hline
 Methods & $p = 512$ & $p = 1024$ & $p = 2048$ & $p = 4096$ & $p = 8192$ & $p = 10000$ \\ \hline
  HDDA & 0.14 (0.02)  & 0.10 (0.02)  & 0.07 (0.01)  & 0.04 (0.01)  & 0.02 (0.01)  &  0.02 (0.01)  \\
  AoYa & 0.14 (0.02)  & 0.11 (0.02)  & 0.07 (0.01)  & 0.04 (0.01)  &  0.02 (0.01)  & 0.02 (0.01)  \\
  ppQDA & 0.08 (0.01)  & {\bf 0.04} (0.01)  & {\bf 0.02} (0.00)  & {\bf 0.00} (0.00)  & {\bf 0.00} (0.00)  & {\bf 0.00} (0.00)  \\
  DA-QDA  & 0.11 (0.09)  & - & - &  - & - & - \\
  IIS-SQDA & 0.10 (0.02)  & 0.10 (0.02)  & 0.07 (0.01) & - & - & -  \\
   RPE-CS & 0.10 (0.01)  & 0.06 (0.02)  &  0.03 (0.01)  & 0.02 (0.01)  &  0.03(0.01)  & 0.03 (0.01)  \\
  RPE-SN & {\bf 0.06} (0.01)  & {\bf 0.04} (0.01)  & 0.03 (0.00)  & 0.03 (0.00)  & 0.03 (0.00)  & 0.02 (0.00)  \\
  RPE-STP & {\bf 0.06} (0.01)  & {\bf 0.04} (0.01)  & 0.04 (0.00)  & 0.03 (0.01)  & 0.04 (0.01)  & 0.04 (0.01)  \\
   \hline
KL$/p$ & 0.04  & 0.04  & 0.04  & 0.03  & 0.03  & 0.03  \\ \hline
\end{tabular}
\caption{Mean and standard deviation (in brackets) of the empirical misclassification probabilities for different methods in Scheme 1.}
\label{tab1}
\end{small}
\end{table}

\vskip5pt
\noindent {\it Results.} 
For DA-QDA and IIS-SQDA, we report results only in low dimensions due to their prohibitive computational cost in high dimensions. In particular, DA-QDA takes around 120 minutes for $p=1024$, while IIS-SQDA requires 30 minutes for $p=2048$.

In Scheme 1, one can observe significant differences in the mean vectors, traces and sum of the off-diagonal elements of the covariance matrices, which increase with dimensionality.
Consequently,  all methods achieve comparable performance in high-dimensional settings, with ppQDA demonstrating the best overall accuracy. Among the RPE-based methods, RPE-SN performs marginally better than the other~two.

%, Additionally, the difference between the mean vectors as well as the KL divergence increases with dimensionality.

\begin{table}[ht]
\centering
    \renewcommand*{\arraystretch}{1.2}
    \begin{small}
\begin{tabular}{| c | c  c  c  c c c |}
  \hline
 Methods & $p = 512$ & $p = 1024$ & $p = 2048$ & $p = 4096$ & $p = 8192$ & $p = 10000$ \\ \hline
  HDDA & {\bf 0.00} (0.00)  & 0.50 (0.02)  & 0.29 (0.05)  & 0.44 (0.03)  & 0.30 (0.04)  & 0.20 (0.04)  \\
  AoYa & 0.50 (0.03)  & 0.50 (0.02)  & 0.50 (0.02)  & 0.50 (0.03)  & 0.50 (0.03)  & 0.50 (0.02)  \\
  ppQDA & 0.48 (0.03)  & 0.49 (0.03)  & 0.49 (0.02)  & 0.49 (0.03)  & 0.49 (0.03)  & 0.49 (0.02)  \\
  DA-QDA & 0.42 (0.08)  & -  & - & - & - & -  \\
  IIS-SQDA & 0.05 (0.02)  & 0.06 (0.02)  & 0.09 (0.02) & - & - & -  \\
  RPE-CS & 0.01 (0.01)  & 0.03 (0.01)  & 0.02 (0.01)  & 0.01 (0.01)  & 0.02 (0.01)  & 0.02 (0.01)  \\
  RPE-SN & {\bf 0.00} (0.00)  & {\bf 0.01} (0.00)  & {\bf 0.00} (0.00)  & {\bf 0.00} (0.00)  & {\bf 0.00} (0.00)  & {\bf 0.00} (0.00)  \\
  RPE-STP & {\bf 0.00} (0.00)  & {\bf 0.01} (0.00)  & {\bf 0.00} (0.00)  & {\bf 0.00} (0.00)  & {\bf 0.00} (0.00)  & {\bf 0.00} (0.00)  \\
   \hline
   $\mathrm{KL}/p$ & 1.05  & 0.60  & 0.89  & 0.64  & 0.72  & 0.65
\\ \hline
\end{tabular}
\caption{Mean and standard deviation (in brackets) of the empirical misclassification probabilities for different methods in Scheme 2.}
\label{tab2}
\end{small}
\end{table}

In Scheme 2, the traces of $\Sigma_1$ and $\Sigma_2$ are identical and the mean vectors coincide. However, both the sum of the off-diagonal elements of $\Sigma_1$ and $\Sigma_2$, along with the KL divergence increases with dimensionality. The absence in mean and trace differences leads to the poor performance of AoYa. Despite of a super-linear growth in the difference of the off-diagonal sums with dimension, ppQDA performs poorly in this setting. In fact, the classification accuracy of both AoYa and ppQDA approaches that of random guessing.
Although HDDA performs somewhat better than both AoYa and ppQDA, its accuracy remains unsatisfactory. 
The primary reason behind this might be the underlying assumption that both the leading eigenvalues and the intrinsic dimensions of the covariance matrices are identical, which is clearly not true here. 
In contrast, the RPE-based methods exhibit superior performance since the KL divergence increases with dimension. both RPE-SN and RPE-STP consistently achieve the lowest average empirical misclassification probabilities throughout.

\begin{table}[h]
\centering
    \renewcommand*{\arraystretch}{1.2}
    \begin{small}
\begin{tabular}{| c | c  c  c  c c c |}
  \hline
 Methods & $p = 512$ & $p = 1024$ & $p = 2048$ & $p = 4096$ & $p = 8192$ & $p = 10000$ \\
  \hline
  HDDA & 0.43 (0.01)  & 0.11 (0.01)  & 0.20 (0.04)  & 0.28 (0.04)  & 0.31 (0.03)  & 0.23 (0.04)  \\
  AoYa & 0.50 (0.02)  & 0.50 (0.03)  & 0.50 (0.02)  & 0.50 (0.02)  & 0.50 (0.02)  & 0.50 (0.02)  \\
  ppQDA & 0.01 (0.00)  & {\bf 0.00} (0.00)  & {\bf 0.00} (0.00)  & {\bf 0.00} (0.00)  & {\bf 0.00} (0.00)  & {\bf 0.00} (0.00)  \\
  DA-QDA & 0.01 (0.03)  & -  & - & - & - & - \\
  IIS-SQDA & {\bf 0.00} (0.00)  & {\bf 0.00} (0.00)  & {\bf 0.00} (0.00) & - & - & - \\
 $\mathrm{RPE\text{-}CS}$ & 0.04 (0.02)  & 0.02 (0.02)  & 0.01 (0.00)  & {\bf 0.00} (0.00)  & {\bf 0.00} (0.00)  & {\bf 0.00} (0.00)  \\
  RPE-SN & 0.08 (0.01)  & 0.05 (0.01)  & 0.04 (0.01)  & 0.03 (0.00)  &  0.02 (0.00)  & 0.02 (0.00)  \\
  RPE-STP & 0.08 (0.01)  & 0.05 (0.01)  & 0.04 (0.01)  & 0.03 (0.01)  & 0.02 (0.01)  & 0.02 (0.01)  \\
   \hline
KL$/p$ & 0.02  & 0.02  & 0.02  & 0.02  & 0.02  & 0.02  \\ \hline
\end{tabular}
\caption{Mean and standard deviation (in brackets) of the empirical misclassification probabilities for different methods in Scheme 3.}
\label{tab3}
\end{small}
\end{table}

While the mean vectors do not differ in Scheme 3, the difference in the traces of the covariance matrices increases linearly with $p$. Under this scenario, the discriminant function of AoYa yields highly positive values, assigning all test points to a single class. Consequently, its empirical misclassification probability remains 50\% across all values of $p$. As the differences in traces, the sums of off-diagonal elements, eigenvalues and related characteristics of $\Sa$ and $\Sigma_{2}$ increase with $p$, this setup is particularly favorable for methods like ppQDA and IIS-QDA.
It is worth noting that the covariance structure in Scheme 3 does not align with our assumption (A2). Nevertheless, the proposed RPE-QDA methods perform quite well in this setting. Among the RPE-based methods, RPE-CS achieves the best overall performance.

\begin{table}[ht]
\centering
    \renewcommand*{\arraystretch}{1.2}
    \begin{small}
\begin{tabular}{| c | c  c  c  c c c|}
  \hline
 Methods & $p = 512$ & $p = 1024$ & $p = 2048$ & $p = 4096$ & $p = 8192$ & $p = 10000$ \\ \hline
  HDDA & {\bf 0.00} (0.00)  & {\bf 0.00} (0.00)  & 0.29 (0.25)  &  0.43 (0.15)  & 0.32 (0.22)  &  0.23 (0.22)  \\
  AoYa & 0.50 (0.03)  & 0.49 (0.02)  & 0.50 (0.03)  & 0.48 (0.02)  & 0.48 (0.03)  & 0.49 (0.02)  \\
  ppQDA & 0.50 (0.02)  & 0.50 (0.02)  & 0.50 (0.03)  & 0.47 (0.03)  & 0.48 (0.02)  & 0.49 (0.02)  \\
  DA-QDA & {\bf 0.00} (0.01)  & -  & -  & - & - & -\\
  IIS-SQDA &  {\bf 0.00} (0.00)  & 0.01 (0.02)  & - & - & - & -  \\
  RPE-CS & 0.01 (0.01)  & 0.01 (0.01)  & 0.04 (0.02)  &  0.05 (0.02)  &  0.11 (0.03)  &  0.13 (0.02)  \\
  RPE-SN & {\bf 0.00} (0.00)  & {\bf 0.00} (0.00) & {\bf 0.00} (0.01)  & {\bf 0.00} (0.02)  & 0.01 (0.01)  & 0.02 (0.01)  \\
  RPE-STP & {\bf 0.00} (0.00)  & {\bf 0.00} (0.00)  & {\bf 0.00} (0.00)  & {\bf 0.00} (0.00)  & {\bf 0.00} (0.00)  & {\bf 0.00} (0.00)  \\
   \hline
KL$/p$ & 1.31  & 1.45  & 1.63  & 1.77  & 1.95 & 2.00  \\ \hline
\end{tabular}
\caption{Mean and standard deviation (in brackets) of the empirical misclassification probabilities for different methods in Scheme 4.}
\label{tab4}
\end{small}
\end{table}

In Scheme 4, the covariance matrix $\Sa$ has the same set of eigenvalues as that of $\Sigma_2$, but with different eigen-directions. Neither the trace nor the sum of the off-diagonal elements captures this structural difference. Moreover, the mean vectors are identical. As a result, both ppQDA and AoYa yield a misclassification rate of about $50\%$. In contrast, both RPE-SN and RPE-STP perform significantly better, as the KL divergence increases with the data dimension. RPE-CS also yields a moderate misclassification probability.

To summarize, RPE-based methods consistently demonstrate good performance across the evaluated schemes. We can also observe that our proposed classifiers yield almost \textit{perfect classification} as $p$ is large, across all schemes.
Clearly, this effectiveness in their performance persists as long as the KL divergence increases with $p$
(recall assumption (\Ath)) regardless of the variation in location parameters, or summary measures of the covariance matrices.

\subsection{Computational Time} 
We conclude this section with a comparison of computational time among the methods studied in this section.
Figure \ref{fig:time} depicts the computational time for all methods in Scheme $4$ across varying data dimensions. The left panel shows the overall time range (excluding DA-QDA and IIS-SQDA), while the right panel highlights finer differences among the fast methods.

To evaluate the scalability of the methods in an uniform way, the computational times are presented without any parallelization. 
All computations are performed on a $1.50$ GHz Intel Core i7 system with $16$ GB of memory and $8$ cores.

\begin{figure}[!h]
\centering
\includegraphics[width=0.48 \textwidth]{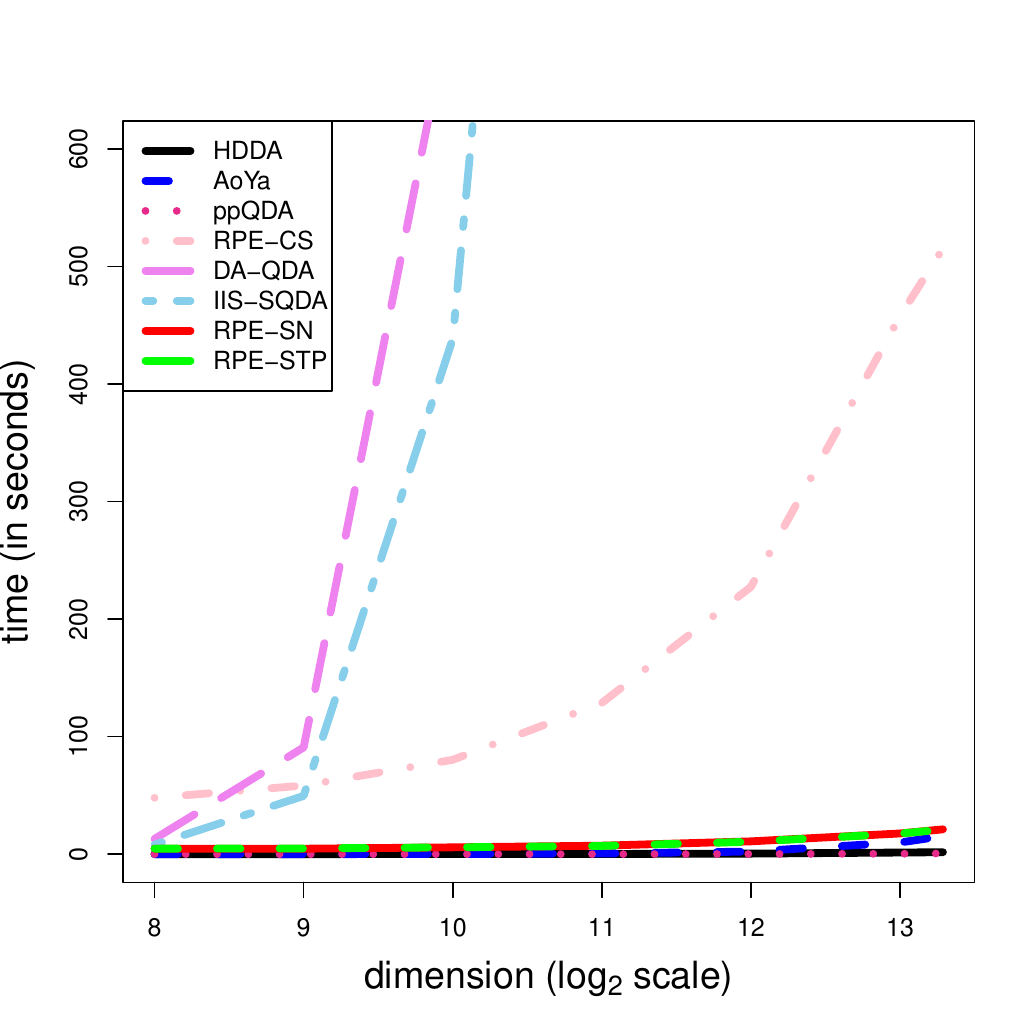}  
\includegraphics[width=0.48 \textwidth]{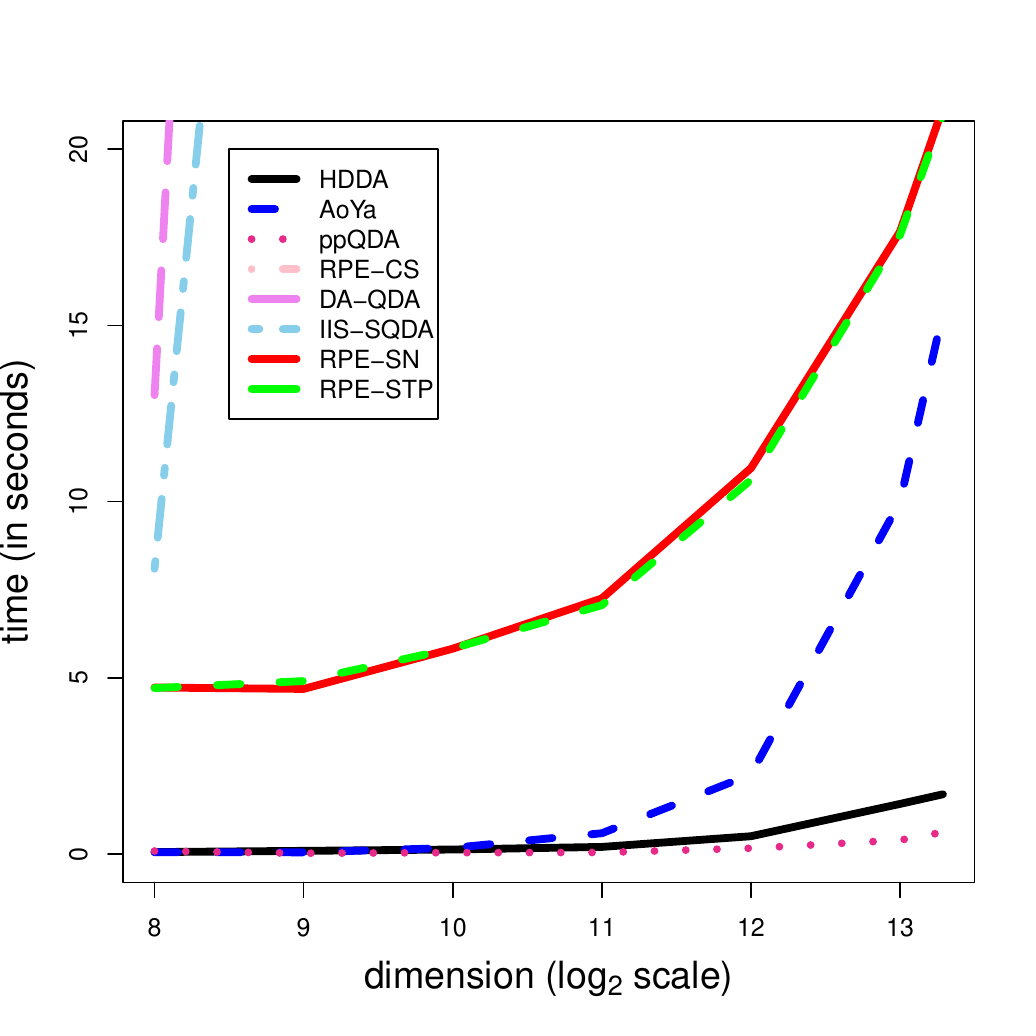} 
\caption{Time comparison for the competing methods in Scheme $4$. Left: full scale; Right: zoomed in to highlight the computationally fast methods.}
\label{fig:time}
\end{figure}

Figure \ref{fig:time} illustrates that the time complexity of DA-QDA and IIS-SQDA increases substantially as the dimensionality of the data rises. In contrast, HDDA, AoYa and ppQDA demonstrate exceptional computational efficiency. RPE-SN and RPE-STP exhibit a marginal increase in computational time compared to AoYA, HDDA and ppQDA. Nevertheless, both RPE-SN and RPE-STP remain significantly more efficient than RPE-CS. Among the competing methods (excluding DA-QDA and IIS-SQDA), RPE-CS shows the steepest increase in time complexity with increasing dimensions. Interestingly, RPE-STP utilizes a sparse choice of random matrices and yet its computational cost is comparable with that of RPE-SN.

\section{Real Data Analysis}\label{sec:realdata-analysis}
We now analyze four data sets related to cancer gene expressions available in the repository of data sets studied in \cite{de2008clustering} (available \href{https://schlieplab.org/Static/Supplements/CompCancer/datasets.htm}{here}) and the Bioinformatics laboratory of University of Ljubljana (available \href{http://www.biolab.si/supp/bi-cancer/projections/}{here}). Each data set is high-dimensional, where the number of genes is in thousands, while the sample sizes are generally fewer than $100$. In view of the small sample sizes, we use leave-one-out cross-validation (LOOCV) to calculate the misclassification rates for each~classifier. 
We first provide a short description of the data~sets.

% \vskip10pt
% \noindent{\bf Chronic Myeloid Leukemia Data.}
% This data set (\href{https://file.biolab.si/biolab/supp/bi-cancer/projections/info/CMLGSE2535.html}{CML}) assesses $28$ chronic myeloid leukemia patients based on their response to imatinib treatment, using data from $12625$ genes. It consists of two classes, namely patients not responding to the treatment (NR) and patients responding to the treatment (R) with $12$ and $16$ samples, respectively. 

\vskip10pt
\noindent{\bf Prognosis Data.}
This data set (\href{https://file.biolab.si/biolab/supp/bi-cancer/projections/info/AMLGSE2191.html}{PRG}) assesses $54$ acute myeloid leukemia patients based on their prognosis after treatment (remission or relapse of disease), using data from $12625$ genes. It consists of two classes, namely patients who achieved complete remission after treatment (REM) and patients who experienced relapse (REL) with $28$ and $26$ samples, respectively. 

\vskip10pt
\noindent{\bf Mixed-lineage Leukemia Data.}
The mixed-lineage leukemia
(\sloppy{\href{https://file.biolab.si/biolab/supp/bi-cancer/projections/info/MLL.html}{MLL}}) data set contains expressions of $12533$ genes for $72$ patients. It includes three diagnostic classes, namely acute lymphoblastic leukemia (ALL), acute myeloid leukemia (AML) and mixed-lineage leukemia (MLL) with $24$, $28$ and $20$ samples, respectively.

\vskip10pt
\noindent{\bf Gastric Cancer Data.}
This data set (\href{https://file.biolab.si/biolab/supp/bi-cancer/projections/info/gastricGSE2685.html}{GC}) contains gene expression signatures for three different gastric tumors, namely, diffuse gastric tumor (DGT), intestinal gastric tumor (IGT) as well as normal gastric tissue (NGT) with sample sizes $5$, $17$ and $8$, respectively. The diagnosis is based on $4522$~genes.

\vskip10pt
\noindent{\bf Brain Tumor Data.}
This data set (\href{https://schlieplab.org/Static/Supplements/CompCancer/Affymetrix/nutt-2003-v1/index.html}{BT}) consists of $50$ samples from four diagnostic classes, namely, classic glioblastomas (CG), classic anaplastic oligodendrogliomas (CO), nonclassic glioblastomas (NG), nonclassic anaplastic oligodendrogliomas (NO) representing different brain tumors of the central nervous system. The diagnosis is based on $12625$ gene expression signatures. Further, there are $14$ samples from each of the classes CG and NG, $7$ samples from the class CO and $15$ samples from the class NO.

\vskip5pt

\begin{table}[ht]
\footnotesize
\centering
\renewcommand*{\arraystretch}{1.2}
\begin{tabular}{| c | c | c  c  c c c c c c|}
  \hline
Data & $(p,n) $ & HDDA & AoYa & ppQDA & RPE & RPE & RPE & RPE & RPE \\ 
 & $J$ &  & &  & -CS & -SN & -SN2 & -STP & -STP2 \\
  \hline
PRG & (12625, 54) & 0.407 & 0.426 & 0.482 & 0.389 & 0.333 & 0.148 & 0.444 & \textbf{0.111} \\
& 2 & (0.067) & (0.067) & (0.068) & (0.066) & (0.064) & (0.048)  & (0.068) & (0.043)\\
\hline
MLL & (12533, 72) & 0.056 & 0.083 & 0.667 & 0.056 & \textbf{0.027} & 0.056 & 0.083 & 0.056 \\
& 3 & (0.027) & (0.033) & (0.056) & (0.027) & (0.019) & (0.027) & (0.033) & (0.027) \\

\hline
GC & (4522, 30) & 0.367 & 0.600 & 0.333 & - & 0.333 & \textbf{0.267} & 0.333 & 0.300 \\
& 3 & (0.088) & (0.089) & (0.086) & - & (0.086) & (0.081) & (0.086) & (0.084) \\

\hline

BT & (12625, 50) & 0.320 & 0.460 & 0.700 & 0.380 & 0.340 & 0.380 & \textbf{0.280} & 0.320 \\
& 4 & (0.066) & (0.070) & (0.065) & (0.069) & (0.067) & (0.069) & (0.063) & (0.066) \\

\hline

\hline

\end{tabular}
\caption{Mean and standard deviation (in brackets) of empirical misclassification probabilities for the competing methods in real data sets.}
\label{table:realdata}
\end{table}

The methods DA-QDA and IIS-SQDA could not be implemented on these data sets in view of their heavy computational burden for such high-dimensional data. The results for RPE-CS have also been excluded for the gastric cancer data set because the corresponding R package (\href{https://cran.rproject.org/web/packages/RPEnsemble/index.html}{RPEnsemble}) does not support implementation of QDA for small class-wise sample sizes. For RPE-CS, we used $d = 5$, $B_1 = 500$ and $B_2 = 50$ (as recommended by \cite{cannings2017random}). For RPE-SN and RPE-STP, we have taken $B = 200$ and two choices of $d$, namely, $2$ and $(n_{\min} \wedge \lceil\log p\rceil )$.  
The results corresponding to the latter choice are indicated as $\mathrm{RPE\text{-}SN2}$ and $\mathrm{RPE\text{-}STP2}$, respectively.

Table \ref{table:realdata} summarizes the performance of the competing methods in analyzing the above data sets. It clearly indicates that RPE-based methods generally outperform HDDA, AoYa and ppQDA across all these real data sets. 
Computation of HDDA involves the estimation of several parameters for each class (see \cite{bouveyron2007high}).
The high misclassification rates of HDDA in some data sets may stem from incorrect estimation of these parameters.
%For some of these real data sets, it incorrectly estimates the parameters and results in an increased misclassification rate.
%For AoYa, the separation between two (or, more) classes is captured only through the mean vectors and the traces of the estimated covariance matrices, and ppQDA additionally incorporates information from the sum of off-diagonal terms. 
Further, the underlying data generating process of the data sets under consideration might be too complex to effectively capture the discriminative information via selected summary measures, resulting in suboptimal performance of AoYa and ppQDA.  

\begin{figure}[h]
    \centering
    \begin{subfigure}[b]{0.45\linewidth}
        \centering
        \includegraphics[width=1\linewidth]{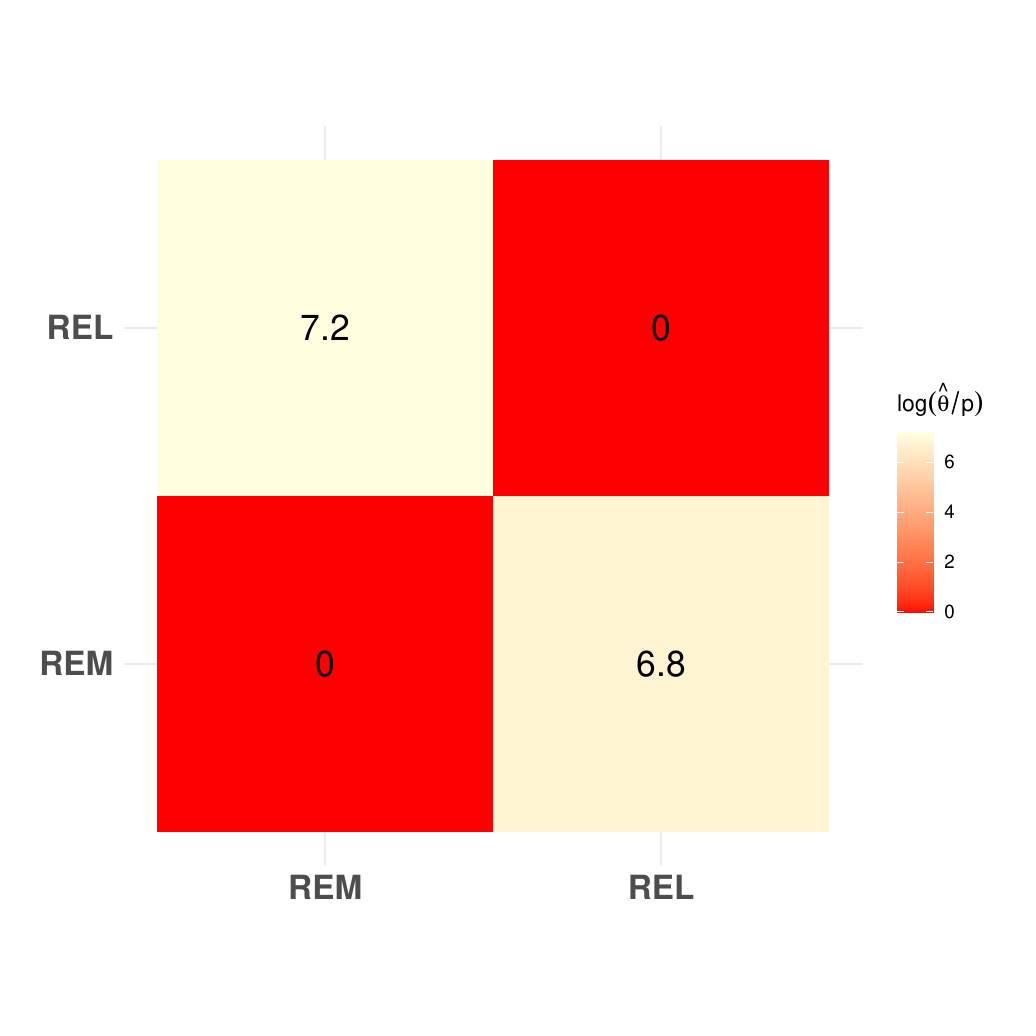}
        \caption{PRG}
        \label{fig:2KLBrain}
    \end{subfigure}
    \hfill
    \begin{subfigure}[b]{0.45\linewidth}
        \centering
        \includegraphics[width=1\linewidth]{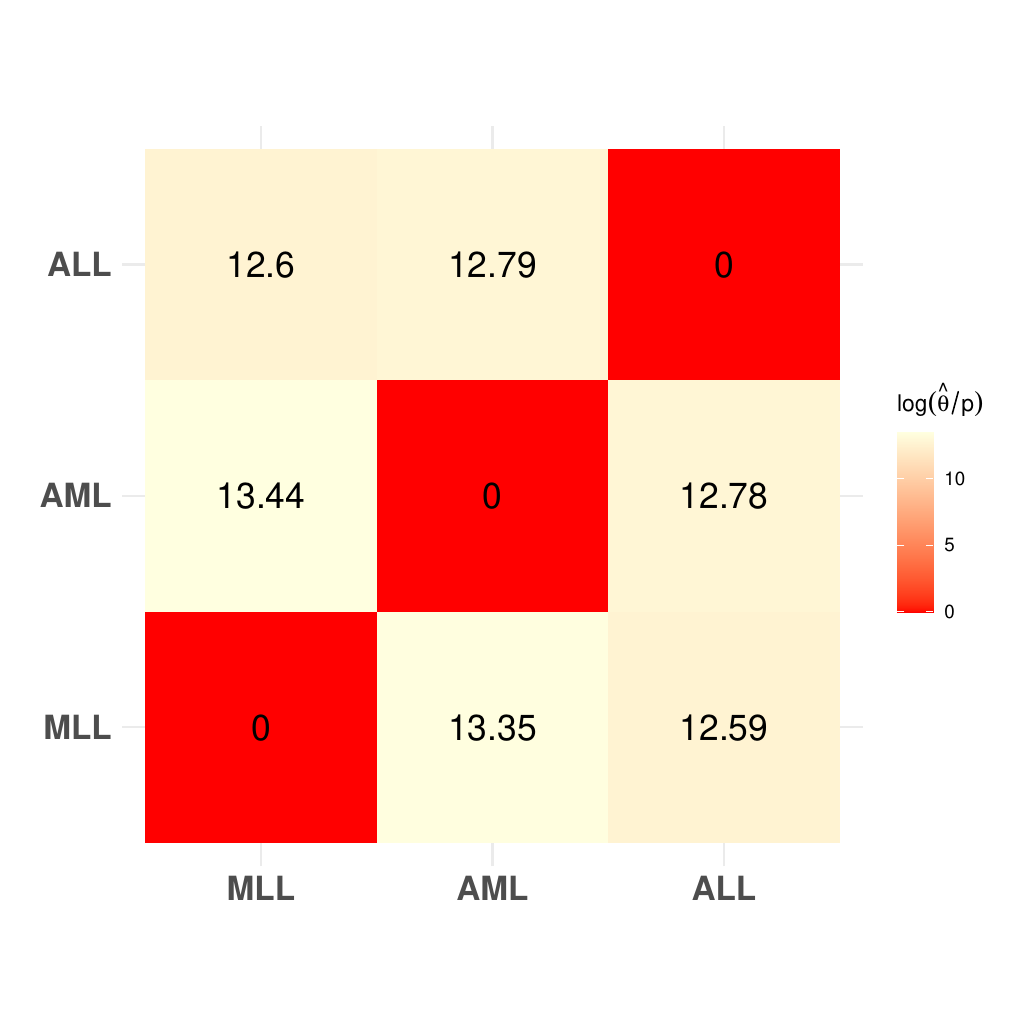}
        \caption{MLL}
        \label{fig:2KLLung}
    \end{subfigure}
    \hfill
    \begin{subfigure}[b]{0.45\linewidth}
        \centering
        \includegraphics[width=1\linewidth]{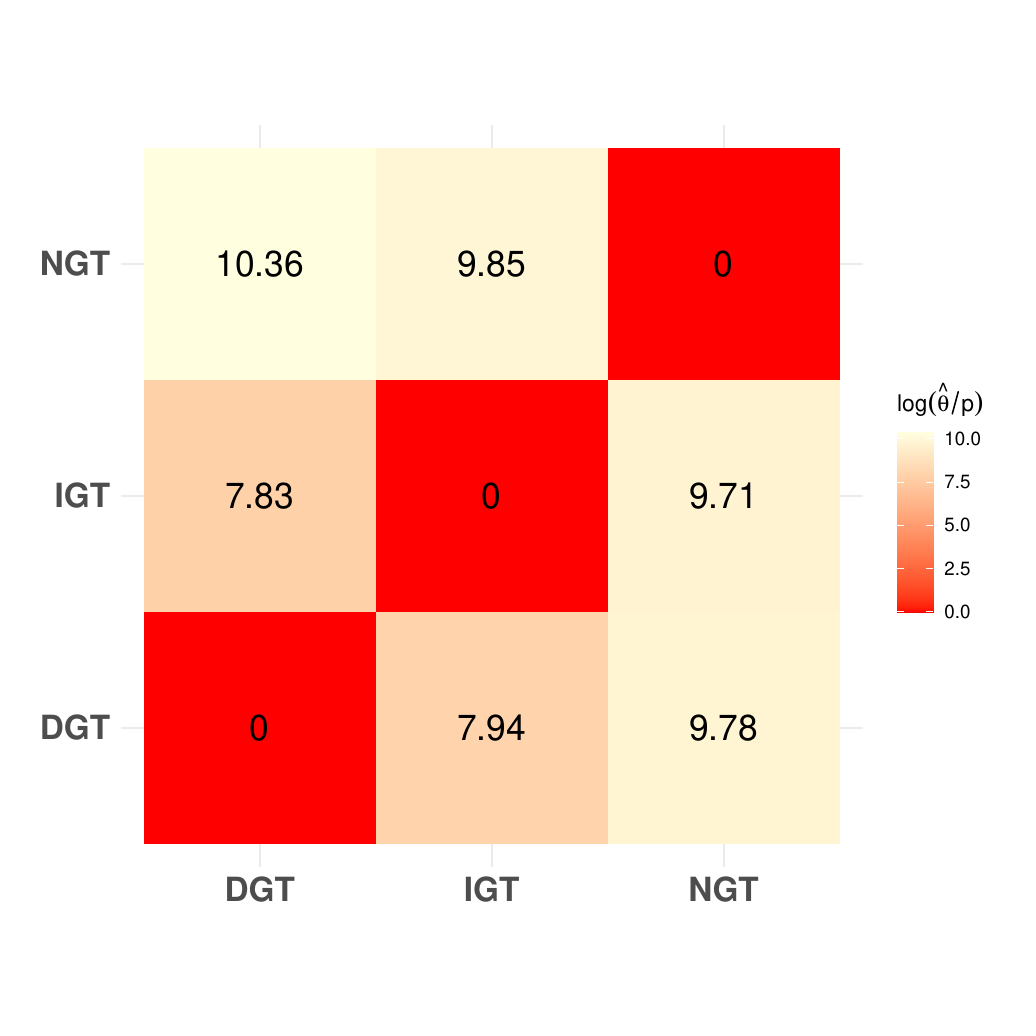}
        \caption{Gastric Cancer}
        \label{fig:2KLMLL}
    \end{subfigure}
    \hfill
    \begin{subfigure}[b]{0.45\linewidth}
        \centering
        \includegraphics[width=1\linewidth]{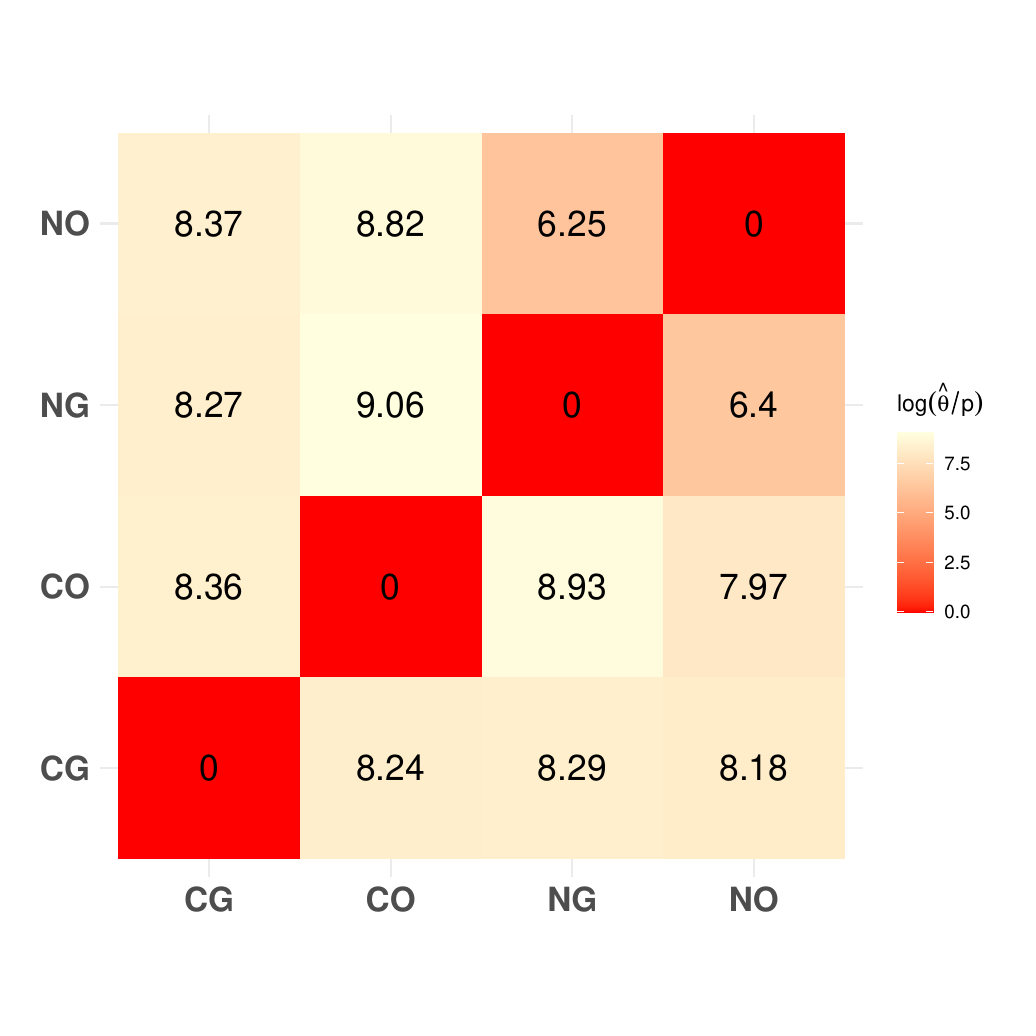}
        \caption{Brain Tumor}
        \label{fig:2KLBreast}
    \end{subfigure}
    \hfill
     \caption{Heatmaps representing the pairwise values of $\log(\widehat{\theta}_{k,\kp}/p)$ for the four real data~sets.}
     \label{fig2:KL}
\vspace{-0.1in}
\end{figure}

% Let us now analyze the real data sets in further detail. To access the separability among the competing populations we estimate the KL divergence among them. 
% Note that, it is computationally prohibitive to estimate the KL divergence in high dimensions by replacing the population parameters by their sample counterparts.  
% Therefore, we derive a lower bound for the KL divergence between classes $k$ and $\kp$, say $\theta_{k,\kp}$, in Section \ref{subsec:App-Numeric},  
% i.e., $\mathrm{KL}_{k,\kp} \geq \theta_{k,\kp}$ for all pairs $(k, \kp)$.
% Next, we provide an estimate of $\theta_{k,\kp}$, denoted by $\widehat{\theta}_{k,\kp}$,
% and present a heatmap of $\log(\hat{\theta}_{k,\kp}/p)$ for each pair of classes $(k, \kp)$ for these real data sets in Figure \ref{fig2:KL}. 

We now examine the real data sets in a finer detail. To assess class separability, we estimate the KL divergence between populations. Note that direct estimation using sample estimates of the parameters is computationally infeasible in high dimensions. Therefore, in Section \ref{subsec:App-Numeric}, we derive a lower bound for the KL divergence between classes $k$ and $\kp$, denoted $\theta_{k,\kp}$, i.e., $\mathrm{KL}_{k,\kp} \geq \theta_{k,\kp}$ for all pairs $(k, \kp)$. We then obtain an estimate of this bound, denoted $\widehat{\theta}_{k,\kp}$, and present a heatmap of $\log(\widehat{\theta}_{k,\kp}/p)$ for each pair of classes in Figure \ref{fig2:KL}.
The plots show that the ratio of $\mathrm{KL}_{k,\kp}$ to $p$ is expected to be large for each pair of classes across all data sets, suggesting a significant separation between them. This aligns well with assumption (\Ath) and justifies the good performance of the methods based on RPE.

\begin{figure}[h]
    \centering
    \begin{subfigure}[b]{0.45\linewidth}
        \centering
        \includegraphics[width=1\linewidth]{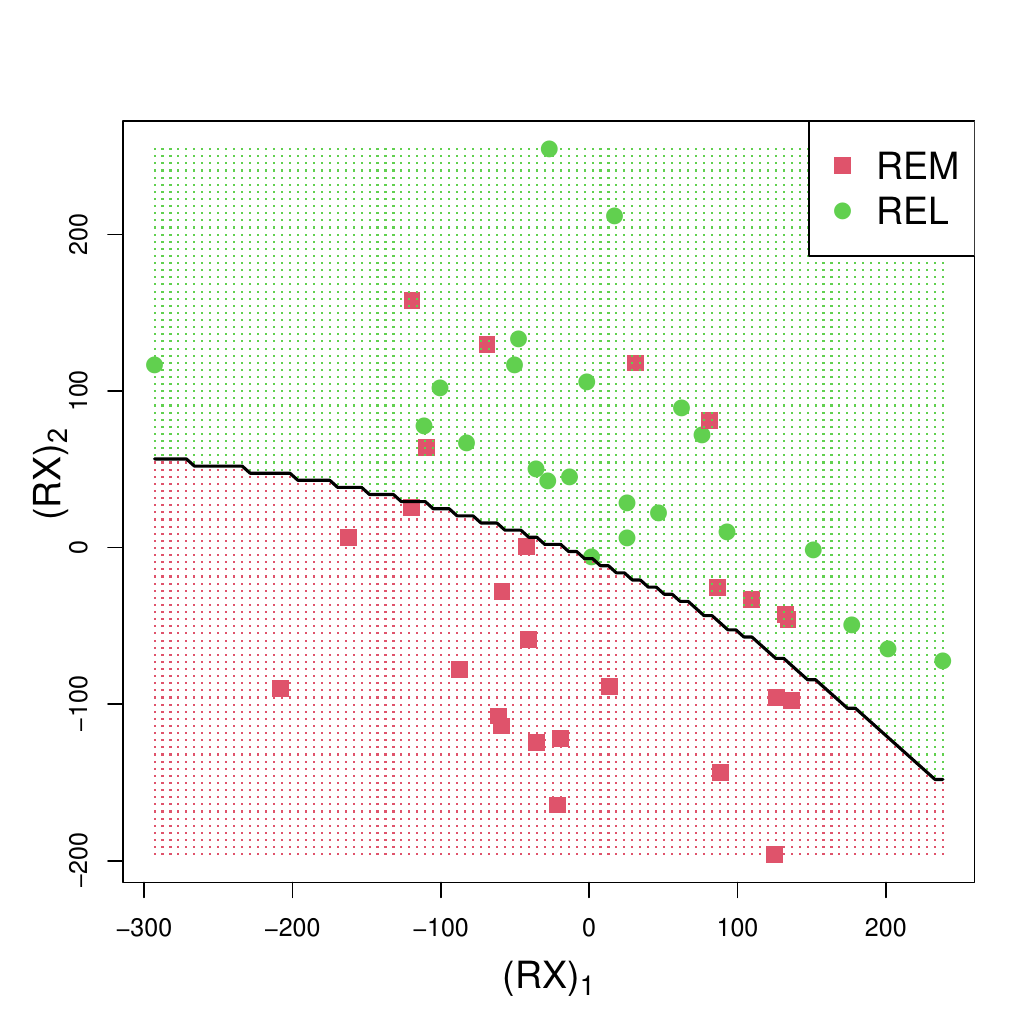}
        \caption{PRG}
        \label{fig:2dimbreast}
    \end{subfigure}
    \hfill
    \begin{subfigure}[b]{0.45\linewidth}
        \centering
        \includegraphics[width=1\linewidth]{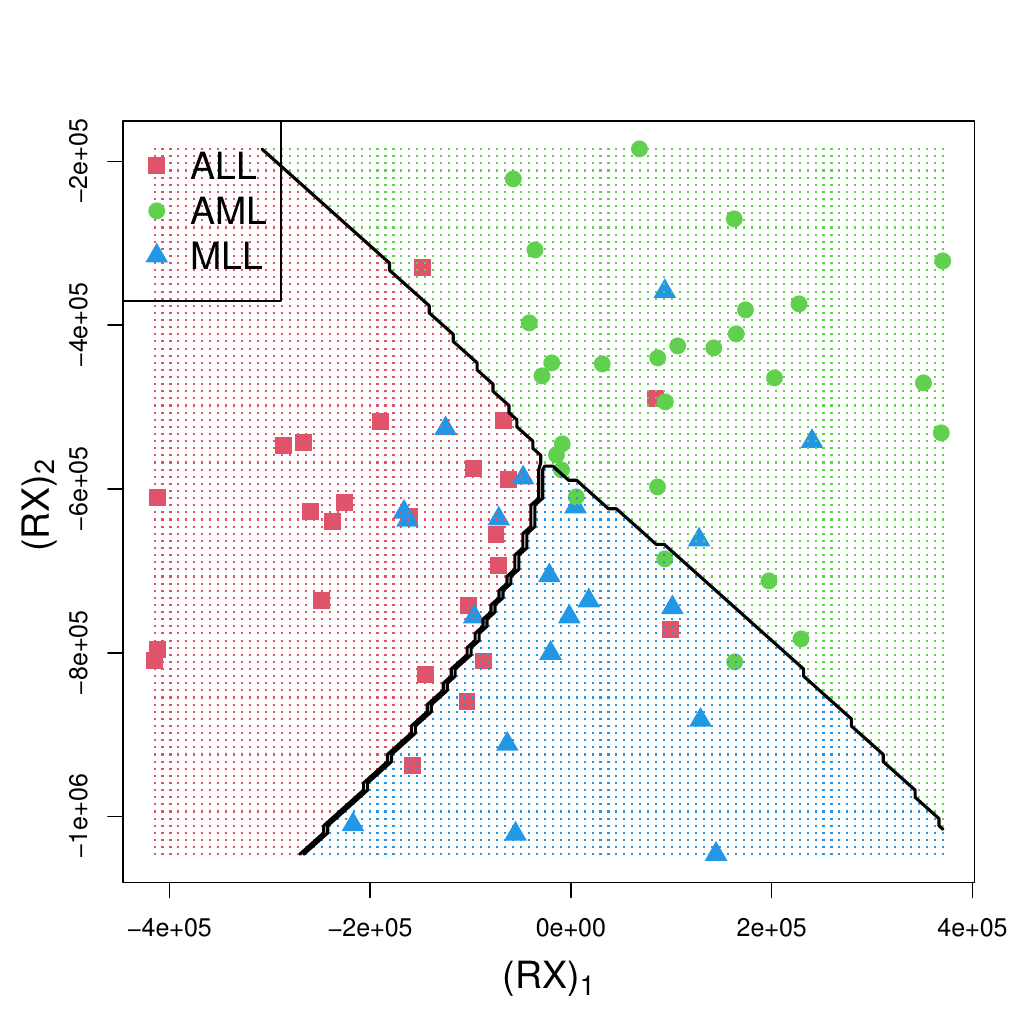}
        \caption{MLL}
        \label{fig:2dimMLL}
    \end{subfigure}
    \hfill
    \begin{subfigure}[b]{0.45\linewidth}
        \centering
        \includegraphics[width=1\linewidth]{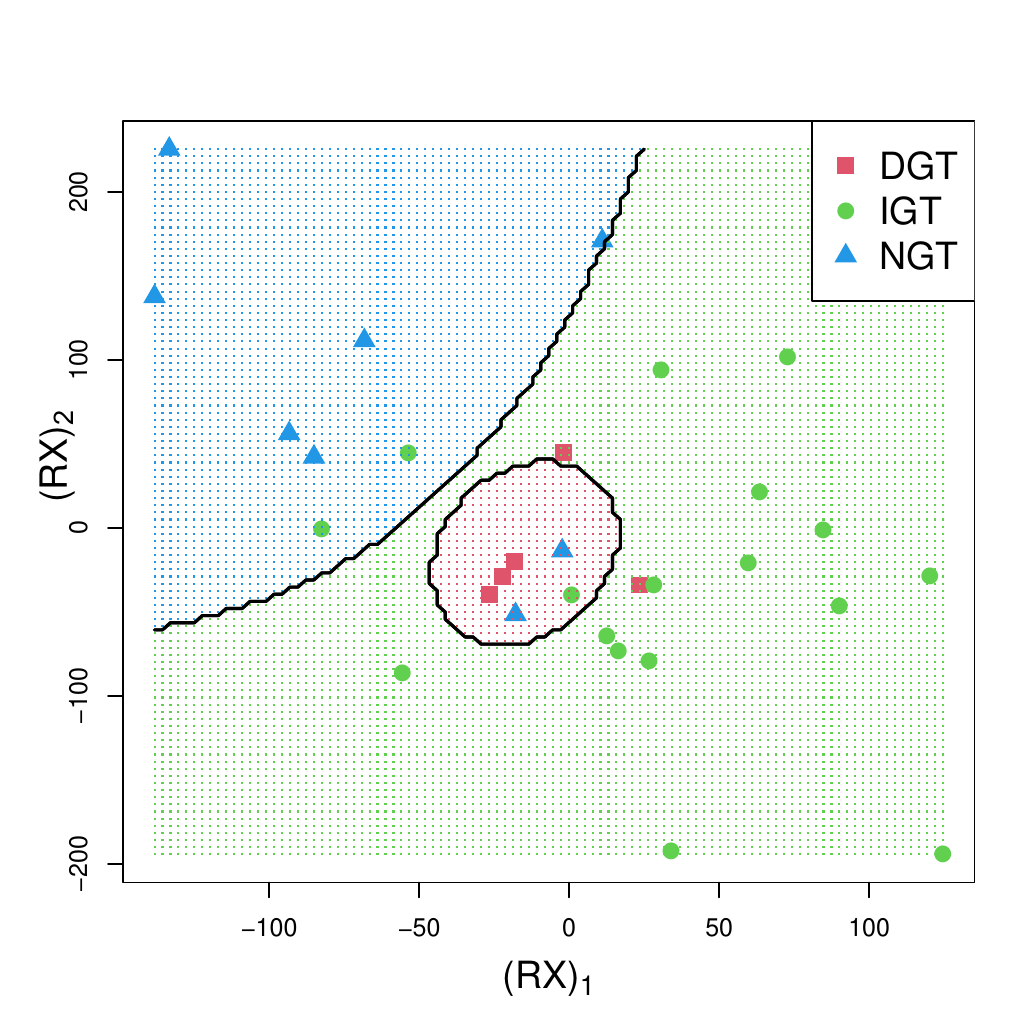}.
        \caption{Gastric Tumor}
        \label{fig:2dimbrain}
    \end{subfigure}
    \hfill
    \begin{subfigure}[b]{0.45\linewidth}
    
        \centering
        \includegraphics[width=1\linewidth]{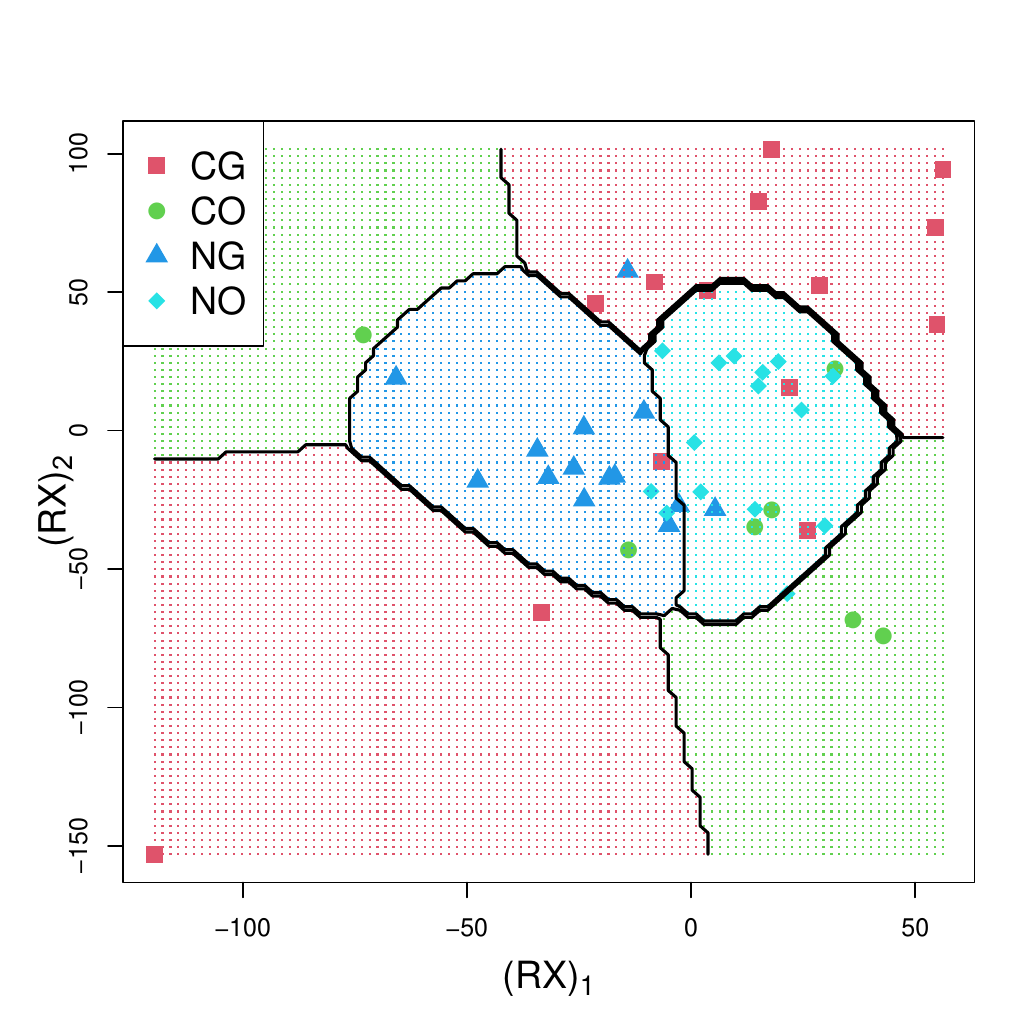}
        \caption{Brain Tumor}
        \label{fig:2dimlung}
    \end{subfigure}
    \caption{QDA boundaries after random projection using SN random matrices in $\mathbb{R}^2$.}
    \label{fig:2dimRP}
    \vspace{-0.1in}
\end{figure}

In Figure \ref{fig:2dimRP}, we project each data onto the two-dimensional plane using a standard normal (SN) random matrix. 
Using one random matrix, we draw the RP-QDA class boundaries (see equation \eqref{eqn:RPDF}) for each data set for visualizations. 
It is evident from Figure \ref{fig:2dimRP} that the classes are quite well separated even in a randomly projected $2$-dimensional subspace. 
% However, this separation is not so clear for the Lung cancer data, as the class AD (comprising of $70\%$ of the training data) dominates the other classes, while class SMCL (with as few as $6$ data points) gets mixed with the other classes. 
Overall, our data analysis establishes quite clearly that methods based on random projections are generally effective for ultrahigh-dimensional classification~problems.

\section{Discussion}

Quadratic Discriminant Analysis (QDA) is a widely used classification method.
%, particularly effective in low-dimensional settings. 
However, its utility diminishes in high-dimensional scenarios due to computational intractability.
To address these challenges, various extensions to classical QDA have been proposed under different structural assumptions on population parameters. For ultrahigh-dimensional settings, we propose a random projection ensemble (RPE) based approach, termed as RPE-QDA. The proposed RPE-QDA method achieves perfect classification as both the data dimension and sample size increase. {\it Perfect classification} is retained even when the dimensionality grows at a sub-exponential rate relative to the sample size. Through extensive simulations and real data analyses, we demonstrate the superior performance of RPE-QDA compared to several existing methods. Notably, in comparisons with another RPE-based method proposed by \cite{cannings2017random}, RPE-QDA achieves comparable classification performance
in significantly less computational time.
%, which is particularly useful in ultrahigh-dimensional~settings.

This study opens several promising avenues for future research. First, ultrahigh-dimensional data often exhibit sparsity, with classification information concentrated in a small subset of variables. In such cases, a targeted random projection approach (see, e.g., \cite{mukhopadhyay2020targeted}), which combines variable screening with RPE, could provide a more tailored solution. Second, while quadratic classifiers are applicable beyond Gaussian populations, our real data analyses suggest that RPE-QDA performs well across a broader range of distributions. Extending our theoretical framework to accommodate for more general distributions would be a valuable contribution. Finally, the RPE framework is highly versatile and definitely not limited to QDA. It could be integrated with other computationally demanding classification methods such as neural networks or random forests in ultrahigh dimensions. Exploring these potential extensions of the RPE methodology presents some interesting directions of future~work.

\section{Appendix}\label{sec:App}
\subsection{Proof of the Theorems}
This part contains the proofs of all the theorems stated in Section \ref{sec:theory}.
%which deal with the misclassification probabilities of various classifiers. 
We first establish the following lemma, which simplifies a multi-class discriminant problem in terms of two class problems. The proof of this lemma is provided in the Supplementary.

\begin{lemma}\label{lm:0}
For the $\delta^{\mathrm{QDA}}$ classifier, the misclassification probability in (\ref{eqn:MP_population}) can be bounded as 
\begin{eqnarray}
 0 \leq \Delta^{\mathrm{QDA}} \leq \sum_{k=1}^{J} \pi_k \sum_{\kp \neq k} \mathbb{P}\left(D_{\kp,k}(\bZ) > 0 \mid \bZ \in P_{k} \right) .
   \label{eqn:twoclass}
\end{eqnarray}
Similar bounds on the misclassification probability can be obtained for the classifiers $\delta^{\text{RPE-QDA}}$ and $\delta_n^{\text{RPE-QDA}}$ by substituting $D_{\kp, k}$ with $D^{\text{RPE}}_{\kp, k}$ (see equation \eqref{eqn:DRPE}) and $\widehat{D}^{\text{RPE}}_{\kp,k}$ (an estimate of $D^{\text{RPE}}_{\kp, k}$), respectively, in equation \eqref{eqn:twoclass}.
\end{lemma}

By Lemma \ref{lm:0}, to show $\Delta^{\mathrm{QDA}} \to 0$ as $p \to \infty$, it is enough to show that $\mathbb{P}\left(D_{k',k}(\bZ) > 0 \mid \bZ \in P_{k} \right) \to 0$, or equivalently, $\mathbb{P}\left(D_{k, \kp}(\bZ) \leq 0 \mid \bZ \in P_{\kp} \right)\to 0$ as $p\to \infty$ for all $(k',k)$.
Since our assumptions are symmetric across all class pairs, we present results for two-class problems only.
%Since our assumptions are symmetric for all pairs of classes, in what follows we will show the results for two-class problems only.

\begin{proof}[Proof of Theorem \ref{thm:0}] 
{\bf Part (a).} In a two class problem, we denote the two populations as $P_{0}$ and $P_{1}$.
Consider $D_{0,1}$ as defined in (\ref{eqn:PDF}) with $(k, \kp) = (0,1)$. For notational simplicity, henceforth, we write $D_{0,1} \equiv D$ by dropping its suffix. Observe that 
\begin{eqnarray}
    \Prob_{\bZ} \left( D(\bZ) \leq 0 \mid \bZ \in P_{0} \right) =  \Prob \left( \bZ^{\prime} \Psi \ \bZ + \boldsymbol{\eta}^{\prime} \bZ + \xi \leq 0 \mid \bZ \sim N_p(\bmu_{0} , \So ) \right), \notag  
\end{eqnarray}
where $\Psi = \Sa^{-1} - \So^{-1}$, $ \boldsymbol{\eta} = 2\left(  \So^{-1} \bmu_{0} -  \Sa^{-1} \bmu_{1} \right) $ and $\xi =  \bmu_{1}^{\top} \Sa^{-1} \bmu_{1} -\bmu_{0}^{\top} \So^{-1} \bmu_{0} + \log \det(\Sa) - \log  \det(\So)   - \log (\pi_{1}) + \log(\pi_{0}) $.
Define $W = 2D(\bZ) =  \bZ^{\prime} \Psi \ \bZ + \boldsymbol{\eta}^{\prime} \bZ + \xi$. It is easy to see that 
\begin{eqnarray}
    \Exp_{\bZ}(W) &=& \tr\left( \So \Sa^{-1}\right) - p + \left(\bmu_{1} - \bmu_{0}\right)^{\top} \Sa^{-1} \left(\bmu_{1} - \bmu_{0}\right) + \log  \det \left( \So^{-1}\Sa \right)   - \log \pi_{1} + \log \pi_{0} \notag \\
    & = & 2 \mathrm{KL}_{0,1}  - \log \pi_{1} + \log \pi_{0} \notag.
\end{eqnarray} 
The variances of $\bZ^{\prime} \Psi \bZ$ and $\boldsymbol{\eta}^{\prime} \bZ$ are
$2 \tr \left(\Psi \So \Psi \So \right) + 4 \bmu_{0}^{\top} \Psi \So \Psi \bmu_{0}$
and $ \boldsymbol{\eta}^{\top}\So \boldsymbol{\eta}$, respectively. 
Observe that $\tr \left(\Psi \So \Psi \So \right) \leq  \tr \left( \Sa^{-1} \So \right)^{2} +p$ and
$  \tr \left( \Sa^{-1} \So \right)^{2}  \lesssim p^{\alpha} \tr \left(  \Sa^{-1} \So \Sa^{-1} \right) \sim p^{1 +\alpha}/ \gamma_{1}^{2}$. Thus, 
$\tr \left(\Psi \So \Psi \So \right)\lesssim  o \left(p^{2} \right)$ holds by assumption (\Ao). 
Further, the second term in the variance of $\bZ^{\prime} \Psi \bZ$ can be analyzed as follows: 
\begin{eqnarray*}
    \bmu_{0}^{\top} \Psi \So \Psi \bmu_{0} & \leq & \bmu_{0}^{\top} \Sa^{-1} \So \Sa^{-1} \bmu_{0} + \bmu_{0}^{\top} \So^{-1} \bmu_{0} \leq \|\bmu_{0} \|^{2} \left\{ \LM \left(  \Sa^{-1} \So \Sa^{-1}  \right) + \LM \left( \So^{-1}\right)\right\} \\
    & \lesssim & p^{\alpha} \|\bmu_{0} \|^{2} = o \big(p^{2} \big)
\end{eqnarray*} 
by using assumptions (\At) and (\Ao).
Similarly, it can be shown that 
$ \boldsymbol{\eta}^{\top}\So \boldsymbol{\eta} = o \left(p^{2} \right) $.
Using the Cauchy-Schwarz (CS) inequality, we get 
$\mathrm{cov}_{\bZ} \left(\boldsymbol{\eta}^{\prime} \bZ , \bZ^{\prime} \Psi \bZ\right) = o\left(p^{2} \right)$.
Combining the facts given above, we have $\mathrm{Var}_{\bZ}(W) =  o\left(p^{2} \right)$.

Thus, $p^{-1} |W - \Exp_{\bZ}(W)| \xrightarrow{P} 0$ as $p\to \infty$. The result now follows since $p^{-1} |\Exp_{\bZ}(W) - 2\mathrm{KL}_{0,1}| \to 0$ when $p\to \infty$ and $W = 2 D(\bZ)$.
\vskip10pt

\noindent {\bf Part (b).} Next, consider the sequence of random variables $V_{p} = W/p$. Since $\pi_{0}$ and $\pi_{1}$ are fixed numbers in $(0,1)$, from the calculations above and assumptions (\At)-(\Ath), we have
$ \lim\inf_{p} \Exp_{\bZ}(V_{p}) \geq \nu_{0}$ and $\lim_{p} \Var_{\bZ}(V_{p}) =0$, where $\nu_0$ is as specified in (\Ath).
Therefore, 
$$ 0 \leq \Prob_{\bZ} (V_{p} \leq 0) \leq \Prob_{\bZ} \left(V_{p} - \Exp_{\bZ}(V_{p})  \leq - \nu_{0}/2 \right),$$
by taking $p\geq p_{\nu}$ with $p_{\nu}$ such that $\Exp(V_{p}) > \nu_{0}/2$ for all $p \geq p_{\nu}$. By Chebyshev's inequality, 
%we now have 
$$ 0 \leq \Prob \left( \left|V_{p} - \Exp_{\bZ}(V_{p}) \right|  \geq  \nu_{0}/2 \right) \leq \dfrac{4 \Var_{\bZ}(V_{p})}{\nu_{0} ^2} \to 0 \quad \text{as}~ p \to \infty.$$ 
Thus, the result follows. 
\end{proof}

\vskip5pt
\noindent The proof of Theorem \ref{thm:1} relies on the following lemmas, whose proofs are provided in the~Supplementary.
%Proof of Theorem \ref{thm:1} is based on the following lemmas. The proofs of these lemmas are deferred to the Supplementary.

\vskip5pt

\begin{lemma}\label{lm:1}
Define $\Psi = \Exp_{R} \big[ R^{\top} \left( R \Sigma R^{\top} \right)^{-1} R \big]$, where $R^{d \times p}=\big(( r_{ij} ) \big)$ with $r_{ij} \stackrel{i.i.d}{\sim} N(0,1)$ for all $(i,j)$. Then,
$  \left\{ p \LM(\Sigma) \right\}^{-1} d\leq \Lm (\Psi) \leq  \LM (\Psi) \leq  \left\{ p \Lm(\Sigma) \right\}^{-1} d. $
\end{lemma}

\begin{lemma}\label{lm:2}
Let assumptions (\At)-(\Ath) be true. 
Define $D_{k,\kp}^{\mathrm{RP}}(\bZ):= \Exp_R \big\{ D_{k,\kp}^{\mathrm{RPE}}(\bZ) \big\}$
for $k, \kp \in \{0,1\}$. If $\bZ \in P_{k}$ and $d = O\left(\log p\right)$ with $d \to \infty$, then with probability tending to \emph{one}, we have
$$ \liminf_{p} \left[ d^{-1} D_{k,\kp}^{\mathrm{RP}}(\bZ) - p^{-1} \mathrm{KL}_{k,\kp} \{ 1 + o(1)\} \right] \geq 0. $$
\end{lemma}

%\vskip5pt

\noindent Next, we prove Theorem \ref{thm:1}.
\begin{proof}[Proof of Theorem \ref{thm:1}]
Using Lemma \ref{lm:0}, it suffices to prove the theorem for a two class problem only. Denote the two populations by $P_{0}$ and $P_{1}$, and $D^{\text{RPE}}_{0, 1}$ and $D^{\text{RP}}_{0,1}$ as in equation (\ref{eqn:DRPE}) and in Lemma \ref{lm:3}, respectively, with $(\kp, k) = (0,1)$. For notational simplicity, we omit the suffix and write $D^{\text{RPE}}_{0,1}$ as $D^{\text{RPE}}$, $D^{\text{RP}}_{0,1}$ as $D^{\text{RP}}$ and so on. %This notation is used throughout the proof.

\vskip5pt
\noindent{\bf Part (a).} First, we verify that the strong law of large numbers (SLLN) holds for the sequence $\left\{ D^{R_{1}} ({\bZ}), D^{R_{2}} ({\bZ}), \ldots \right\}$ for a realization ${\bZ}$ with fixed dimension $p$. The answer is affirmative, since the random variables in the sequence are {i.i.d.} and the common expectation exists. Using Lemma \ref{lm:1} and assumptions (\At)-(\Ath), we get
   \begin{multline}
     \Exp_{R}  \left[ \left({\bZ} - \bmu_{1} \right)^{\top} R^{\top} \big( R \Sa R^{\top} \big)^{-1} R \left({\bZ} - \bmu_{1} \right) -  \left({\bZ} - \bmu_{0} \right)^{\top} R^{\top} \big( R \So R^{\top} \big)^{-1} R \left({\bZ} - \bmu_{0} \right) \right] \\
       \leq p^{-1} d \left( \gamma_{1}^{-1} \left\| {\bZ} - \bmu_{1} \right\|^{2} - p^{-\alpha} \left\| {\bZ} - \bmu_{0} \right\|^{2} \right) . \notag
   \end{multline}  
Further, as $RR^{\top} \sim \mathrm{Wishart}(I_p, p)$, $\Exp_{R} \left[ \log  \det \left(R  R^{\top} \right) \right] = \sum_{j=1}^{d} \Exp_{R} \left(\log U_{j}\right)$, where $U_{j} \sim \chi^{2}_{p - j + 1}$ (which is finite) for $j=1,\ldots,p$. Thus,
\begin{multline}
    \Exp_{R} \left[ \log  \det \big( R \So R^{\top} \big) - \log \det \big(R \Sa R^{\top} \big)  \right] \\
    \leq d \left\{  \log \LM(\So) - \log \Lm(\Sa)  \right\} \leq  d\left\{ \alpha \log p - \log \gamma_{1} \right\}.\notag 
\end{multline}
Similarly, one can provide a finite lower bound to the expectation above.
Thus, the expectation of $D^{R_b} ({\bf z})$ is finite and SLLN is applicable. The almost sure limit is $D^{\text{RP}}({\bZ}) \coloneqq \Exp_{R_{b}} \left[D^{R_b} ({\bZ}) \right]$.

\vskip5pt
\noindent{\bf Part (b).}  The proof of this part follows directly from Lemma \ref{lm:2}.
    
\vskip5pt
\noindent{\bf Part (c).} 
By part (a), we have 
$$ \lim_{B\to \infty} \Prob_{\bf R} \left( \left| D^{\text{RPE}}({\bZ}) - D^{\text{RP}}({\bZ})\right| > \nu_{0} /4 \right) = 0. $$
This implies that for each ${\bf z}$, the indicator function 
$h_{\bf R}({\bf z}) = \mathbb{I} \left( \left| D^{\text{RPE}}({\bf z}) - D^{\text{RP}}({\bf z}) \right| > \nu_{0} /4 \right) $ converges in probability to $0$ in ${\bf R} := (R_{1}, \ldots, R_{B})$.
For each fixed $p$, using the dominated convergence theorem (DCT), we have
\begin{eqnarray*}
\lim_{B\to \infty} \Prob_{\bZ,{\bf R}} \left( \left| D^{\text{RPE}}(\bZ) - D^{\text{RP}}(\bZ)\right| > \nu_{0}/4 \right) 
&=& \lim_{B\to \infty} \Exp_{\bZ,{\bf R}}\left(  h_{\bf R}(\bZ) \right) \\
&=& \Exp_{\bZ} \left[\lim_{B\to \infty} \Exp_{{\bf R} \mid \bZ}\left(  h_{\bf R}(\bZ) \right) \right]= 0. 
\end{eqnarray*}
Next, we will consider the decision rule given by $D^{\text{RP}}(\bZ)$. By Part (b) and (\Ath) we have 
\begin{eqnarray}
    2 d^{-1}  D^{\text{RP}}(\bZ) \geq  2 p^{-1} \text{KL}_{0,1} \left\{1+o(1)\right\} + o_{p}(1)
     \geq \nu_{0} + o(1)  > 2^{-1} \nu_{0}. \label{eqn:lm3}
\end{eqnarray}   
 Let $\Delta^{\text{RP}}$ denote the misclassification probability obtained by replacing $\delta$ with $\delta^{\text{RP}}$ in equation (\ref{eqn:MP_population}), where $\delta^{\text{RP}}$ is the classifier corresponding to the decision rule $D^{\text{RP}}$. 
From (\ref{eqn:lm3}), we have 
\begin{eqnarray}
\Prob_{\bZ} \left(D^{\text{RP}} (\bZ) \leq 0 \mid \bZ \in P_{0} \right) \wedge   \Prob_{\bZ} \left(D^{\text{RP}} (\bZ) > 0 \mid \bZ \in P_{1} \right) \to 0\label{eqn:lm2}
\end{eqnarray}
as $d \to \infty$ (or, $p \to \infty$).
Now, by denoting $\Prob_{\bZ, \bm{R}}$ as $\Prob$, we have
\begin{multline}
    \Prob \left(D^{\text{RPE}} (\bZ) \leq 0 \mid \bZ \in P_{0} \right) \leq  \Prob \left(
   D^{\text{RPE}} (\bZ) \leq 0 \mid \left|D^{\text{RP}} (\bZ) - D^{\text{RPE}} (\bZ) \right| < \nu_{0}/ 4 , ~ \bZ \in P_{0}  \right) \notag \\
    \qquad \qquad +~ \Prob \left( \left| D^{\text{RPE}}(\bZ) - D^{\text{RP}}(\bZ)\right| > \nu_{0}/4 \right) \notag \\
 \leq  \Prob \left(
   D^{\text{RP}} (\bZ) < \nu_{0} /4 \mid  \bZ \in P_{0}  \right)
   + \Prob \left( \left| D^{\text{RPE}}(\bZ) - D^{\text{RP}}(\bZ)\right| > \nu_{0} /4 \right). \notag
\end{multline}
Taking limit over $B$ first, we get
$$ \lim_{B\to \infty} \Prob \left(D^{\text{RPE}} (\bZ) \leq 0 \mid \bZ \in P_{0} \right) \leq 
\Prob \left(   D^{\text{RP}} (\bZ) <  \nu_{0} /4 \mid  \bZ \in P_{0}  \right)$$
almost surely in $\bZ$. From equation (\ref{eqn:lm2}), the probability stated above converges to $0$ as $p\to \infty$. 

Proceeding similarly, one can also show that
$  \lim_{p\to \infty} \lim_{B\to \infty} \Prob \left(D^{\text{RPE}} (\bZ) \geq 0 \mid \bZ \in P_{1} \right) = 0.$
Combining the above two results, we get $\Delta^{\text{RPE}}\xrightarrow{} 0$ as $B\to \infty$ and $p\to \infty$.
\end{proof}

\vskip5pt
The proof of Theorem \ref{thm:2} is based on the following lemma, whose proof is deferred to the Supplementary. Let the sample sizes of populations $P_0$ and $P_1$ be $n$ and $m$, respectively.

\begin{lemma}\label{lm:3}
Let assumptions (\At)-(\Af) be true.  Define $\delhatrp (\bZ) =  \Exp_{R} \left[\widehat{D}^{\text{RPE}}_{n}(\bZ) \right]$, where $\widehat{D}^{\text{RPE}}_{n}(\bZ)$ is obtained from $D^{\text{RPE}}(\bZ)$ (see (\ref{eqn:DRPE})) by substituting $(\pi_{k} , \bmu_{k}, \Sigma_{k})$ with $(\widehat{\pi}_{k} , \widehat{\bmu}_{k}, \widehat{\Sigma}_{k})$ for $k = 0,1$ as defined in (\ref{eqn:estimates}). For both cases (a) $\bZ \in P_0$ and (b) $\bZ \in P_1$, we have
    $$d^{-1} \left\{ \delhatrp (\bZ) - D^{\text{RP}}(\bZ) \right\} \xrightarrow{P} 0 \quad \text{as } \min\{m, n, p\} \to \infty.$$ 
\end{lemma}

\noindent Finally, we prove Theorem \ref{thm:2} below.

\vskip5pt
\begin{proof}[Proof of Theorem \ref{thm:2}]

Following Lemma \ref{lm:0} and under a balanced design, it is sufficient to establish the theorem for a two class problem by denoting the populations as $P_{0}$ and $P_{1}$.
The key quantity here is $\widehat{D}^{\text{RPE}}_{0, 1}$ which is the estimated version of $D^{\text{RPE}}_{0, 1}$ (see (\ref{eqn:DRPE})) with $(\kp, k) = (0,1)$. For notational simplicity, we omit the suffix and henceforth, write $\hat{D}^{RPE}_{0,1}$ as $\hat{D}^{RPE}$.

Let $\bZ \in P_{0}$. Denote the samples of two populations as $X_{n} = [ \bX_{1},\ldots,\bX_{n}]^{\top}$ and $Y_{m} = [\bY_{1},\ldots,\bY_{m}]^{\top}$.
Consider the following
  \begin{multline}
      \Prob \left( \delhatrpe (\bZ)  \leq 0 \mid \bZ \in P_{0} \right) 
      \leq \Prob \left( \frac{1}{d}\left| \delhatrpe (\bZ) - D^{\text{RPE}} (\bZ) \right|  > \frac{\nu_{0}}{8} \mid \bZ \in P_{0} \right) \\ 
      + \Prob \left( \frac{1}{d}\left| D^{\text{RP}} (\bZ) - D^{\text{RPE}} (\bZ) \right|  > \frac{\nu_{0}}{8} \mid \bZ \in P_{0} \right)
      +  \Prob \left( \frac{D^{\text{RP}} (\bZ)}{d} < \frac{\nu_{0}}{8} \mid \bZ \in P_{0} \right), \label{eqn:thm3_1}
  \end{multline}
  where the probabilities are taken with respect to the joint distribution of $(\bZ\in P_{0}, X_n \in P_{0}, Y_m \in P_{1}, {\bf R} )$.
Using the proofs of Theorem \ref{thm:1} and Lemma \ref{lm:2}, we can show that the last two terms in (\ref{eqn:thm3_1}) converge to \emph{zero} as $B\to \infty$, followed by $p\to \infty$.

The random variable in the first component of equation (\ref{eqn:thm3_1}) can be expressed as
\begin{multline}
     \delhatrpe (\bZ) - D^{\text{RPE}} (\bZ) \\
     = \left\{ \delhatrpe (\bZ) - \delhatrp (\bZ) \right\} + \left\{\delhatrp (\bZ) -D^{\text{RP}} (\bZ) \right\} 
     + \left\{D^{\text{RP}} (\bZ) -D^{\text{RPE}} (\bZ) \right\}. \notag
\end{multline}
We denote $\mathbb{P}_{\bZ, X_n, Y_m, {\bf R}}$ as $\mathbb{P}$. 
The first term in equation (\ref{eqn:thm3_1}) can be bounded above by
\begin{multline}
\Prob \left( \frac{1}{d}\left| \delhatrpe (\bZ) - \delhatrp (\bZ) \right|  > \frac{\nu_{0}}{24} \mid \bZ \in P_{0} \right) + \Prob \left( \frac{1}{d}\left| \delhatrp (\bZ) - D^{\text{RP}} (\bZ) \right|  > \frac{\nu_{0}}{24} \mid \bZ \in P_{0} \right)  \\
+ \Prob \left( \frac{1}{d}\left| D^{\text{RPE}} (\bZ) - D^{\text{RP}} (\bZ) \right|  > \frac{\nu_{0}}{24} \mid \bZ \in P_{0} \right).
    \label{eqn:thm3_2}
\end{multline}
The last term in (\ref{eqn:thm3_2}) is similar to the second term in (\ref{eqn:thm3_1}), and can be shown to converge to zero as in the proof of Theorem \ref{thm:1}.

Similar to the proof of Theorem \ref{thm:2}, it can be shown that conditional on $(\bZ, X_n, Y_m)$, we~get
$$ \lim_{B\to \infty} \Prob_{{\bf R}|\bZ, X_n, Y_m} \left( \frac{1}{d}\left| \delhatrpe (\bZ) - \delhatrp (\bZ) \right|  > \frac{\nu_{0}}{24} \right) = 0. $$
Therefore, the function $g_{\bf R} (\bZ, X_n, Y_m) = \mathbb{I} \left( d^{-1} \left| \delhatrpe (\bZ) - \delhatrp (\bZ) \right|  > \nu_{0}/ 24 \right)$ converges in probability to $0$ (with respect to the distribution of ${\bf R}$ given $(\bZ, X_n, Y_m)$). Using DCT (for each fixed $p$), we have
\begin{multline}
\lim_{B\to \infty} \Prob_{(\bZ, X_n, Y_m, {\bf R})} \left( \left| \delhatrpe (\bZ) - \delhatrp (\bZ) \right| > \nu_{0}/24 \right) 
= \lim_{B\to \infty} \Exp_{(\bZ, X_n, Y_m, {\bf R})}\left[  g_{{\bf R}}(\bZ, X_n, Y_m) \right] \\
= \Exp_{\bZ, X_n, Y_m} \left[\lim_{B\to \infty} \Exp_{({\bf R} \mid \bZ, X_n, Y_m)}\left[ g_{{\bf R}} (\bZ, X_n, Y_m) \right] \right]= 0. \notag
\end{multline}
Thus, the limit with respect to $p$ is zero as well. 

Finally, we consider the middle term in (\ref{eqn:thm3_2}). Using Lemma \ref{lm:3}, this probability also converges to zero as $n \to \infty$. By equation (\ref{eqn:thm3_1}), the result now follows when $\bZ\in P_{0}$. 

Similarly, we can show that $\Prob_{(\bZ, X_n, Y_m, {\bf R})} \left( \delhatrpe (\bZ)  \geq 0 \mid \bZ \in P_{1} \right) \to 0$ as $\min\{n,p,B\} \to \infty$. These two results together prove $\Delta_n^{\text{RPE-QDA}} \to 0$ as $\min\{n,p,B\} \to \infty$.
\end{proof}

\subsection{Additional Details Related to Section \ref{sec:realdata-analysis}} \label{subsec:App-Numeric}

% {\color{red}
In this sub-section, we describe the quantities depicted in Figure \ref{fig2:KL} of Section \ref{sec:realdata-analysis}. 
The heatmaps demonstrate the reason behind the good performance of RPE-QDA methods in the gene expression data sets. One of the key assumptions of our theory is the separability assumption (\Ath), which requires the KL divergence between two competing populations be at least of the order of $p$.
Towards that, we first obtain a lower bound of the KL divergence and then estimate it using the training data, as described below.
% }

The KL divergence between any two populations $P_k$ and $P_{\kp}$ for any $k, \kp \in \{1, \ldots, J\}$ is 
\begin{eqnarray}
   2\mathrm{KL}_{k, \kp} &=& tr(\Sk^{-1}\Sigma_{\kp}) + (\bmu_{k} - \bmu_{\kp})\Sk^{-1}(\bmu_k - \bmu_{\kp}) - p + \log\det(\Sk \Sigma_{\kp}^{-1}) \notag \\
   & \geq & \sum_{j=1}^{p} \left\{ \lambda_{j}  (\Sk^{-1}\Sigma_{\kp}) - \log \lambda_{j} (\Sk^{-1}\Sigma_{\kp}) - 1 \right\}  + (\bmu_{k} - \bmu_{\kp})(I_p + \Sk)^{-1}(\bmu_k - \bmu_{\kp}) \notag\\
   & \geq &  (\bmu_{k} - \bmu_{\kp})(I_p + \Sk)^{-1}(\bmu_k - \bmu_{\kp}), \label{eqn:KLLB}
\end{eqnarray}
where $\lambda_j(A)$ denotes the $j^{\text{th}}$ highest eigenvalue of the matrix $A$ for $j=1,\ldots,p$. Here, the first inequality stems from the fact that $(I+ \Sk)^{-1} \geq \Sk^{-1}$ (i.e., ${\bf y}^{\top} (I+ \Sk)^{-1} {\bf y} \leq  {\bf y}^{\top} \Sk^{-1} {\bf y}$ for any ${\bf y}$) and the second inequality is due to the fact that $h(x) = x - \log x - 1 \geq 0$ for all~$x>0$.

Define the quantity in equation (\ref{eqn:KLLB}) as $2\theta_{k, \kp}$, and estimate it using the sample means and covariance matrices. The sample estimate of $\theta_{k,\kp}$, denoted by $\widehat{\theta}_{k,\kp}$, is
%\begin{eqnarray}
 $ \widehat{\theta}_{k, \kp} =   (\widehat{\bmu}_{k} - \widehat{\bmu}_{\kp})(I + \widehat{\Sigma}_{k})^{-1}(\widehat{\bmu}_{k} - \widehat{\bmu}_{\kp}) / 2$,
%\end{eqnarray}
where $\widehat{\bmu}_{k}$ and $\widehat{\Sigma}_{k}$ are given by equation (\ref{eqn:estimates}).
In Figure \ref{fig2:KL}, we present the values of $\log \big(\widehat{\theta}_{k, \kp}/ p \big)$ for all pairs $(k, \kp)$. 
Clearly, a high positive value of this quantity will indicate a high value for~$\mathrm{KL}_{k, \kp}/p$.
This in turn illustrates the validation of assumption (\Ath).

\vskip5pt

\noindent{\bf Supplementary.}
The Supplementary material (available \href{https://www.dropbox.com/scl/fi/eu7zojoxj48i82xx8gnhj/Supplementary.pdf?rlkey=mfvbzn12wncn0saqk8eanwpej&st=21t6gs39&dl=0}{here}) contains the proofs of Lemmas \ref{lm:0}-\ref{lm:3}. % that we have used in this paper. 

%\newpage
\newpage
\bibliography{manuscript}
\bibliographystyle{apalike}

% \begin{document}
% \maketitle
% \def\theequation{S.\arabic{equation}}

\end{document}